\documentclass[11pt]{article}
\usepackage{lineno,hyperref}
\usepackage{bbm}
\usepackage{float}
\usepackage{mathrsfs}
\usepackage{graphicx}
\usepackage[utf8]{inputenc}
\usepackage{amsmath}
\usepackage{amstext}
\usepackage{amsfonts}
\usepackage{amsthm}
\usepackage{amssymb}
\usepackage[bf,SL,BF]{subfigure}
\usepackage{subfigure}
\usepackage{caption}
\newcommand{\triple}{{\vert\kern-0.25ex\vert\kern-0.25ex\vert}}
\theoremstyle{plain}
\newtheorem{definition}{Definition}[section]
\newtheorem{theorem}[definition]{Theorem}
\newtheorem{lemma}[definition]{Lemma}

\newtheorem{assumption}[definition]{Assumption}
\theoremstyle{definition}

\begin{document}
\title{\bf Delay stochastic interest rate model with jump and strong convergence in Monte Carlo simulations}
\author
{{\bf Emmanuel Coffie \footnote{Email: emmanuel.coffie@strath.ac.uk}}
\\[0.2cm]
Department of Mathematics and Statistics,
 \\[0.2cm]
University of Strathclyde,  Glasgow G1 1XH, U.K.}
\date{}
\maketitle
\begin{abstract}
 In this paper, we study analytical properties of the solutions to the generalised delay Ait-Sahalia-type interest rate model with Poisson-driven jump. Since this model does not have explicit solution, we employ several new truncated Euler-Maruyama (EM) techniques to investigate finite time strong convergence theory of the numerical solutions under the local Lipschitz condition plus the Khasminskii-type condition. We justify the strong convergence result for Monte Carlo calibration and valuation of some debt and derivative instruments.

\end{abstract}
\section{Introduction}
\noindent Despite of the popularity of several asset price stochastic models such as Black-Scholes (1973) \cite{blackshole}, Merton (1973) \cite{merton}, Vasicek (1977) \cite{vasicek}, Dothan (1978) \cite{dothan}, Brennan and Schwartz (1980) \cite{brennan}, Cox, Ingersoll and Ross (CIR) (1985) \cite{cox3} and Lewis (2000) \cite{lewis}, they may not be well-specified adequately to fully explain certain types of empirical phenomena in most financial markets. For instance, volatility 'skews' and 'smiles', and tail distribution of asset prices which have been observed empirically from various sources of financial data, may not be captured by these models (e.g., see \cite{mao3, kou,espen}).
\par
In recent times, several interesting research works have been directed towards adequate explanation of dynamical behaviours of financial variables against unexpected occurrences of these empirical phenomena. For instance, contrary to efficient market hypothesis, the delayed GBM \cite{mao3}, CIR \cite{wu3} and CEV \cite{cev} models have been introduced as extensions of \cite{blackshole}, \cite{cox3} and \cite{lewis} to incorporate volatility 'skews' and 'smiles' based on non-Markovian property to explain asset price dynamics. Similarly, a variety of jump diffusion models have also been proposed to explain jump behaviour or tails of distribution of asset prices. For references, see, for example, Merton (1976) \cite{jmerton}, Lin and Yeh (1999) \cite{linyeh} , Kou (2002) \cite{kou} and Wu et al. (2008) \cite{wu2}.
\par
Ait-Sahalia model proposed in \cite{aitsahalia} serves extensively as an indispensable tool for capturing dynamics of term structure of interest rates. This model is driven by a highly nonlinear stochastic differential equation (SDE)
\begin{equation}\label{eq:1}
 dx(t)=(\alpha_{-1}x(t)^{-1}-\alpha_{0}+\alpha_{1}x(t)-\alpha_{2}x(t)^{2})dt+\sigma x(t)^{\theta}dB(t),
\end{equation}
$x(0)=x_0$, for any $t>0$, where $\alpha_{-1},\alpha_{0},\alpha_{1}, \alpha_{2}$ are positive constants and $\theta >1$. Besides interest rates, it has also been considerably used to explain dynamics of asset price, volatility and other financial instruments. There have been several rich literature concerning with this model. For instance, Cheng (2009) in \cite{cheng} studied this model and established weak convergence of EM scheme. Szpruch et al. (2011) in \cite{Szpruch} generalised this model and established strong convergence of implicit EM method as well as preservation of positive approximate solutions of this method when a monotone condition is fulfilled. Dung (2016) in \cite{dung} derived explicit estimates for tail probabilities of solutions to the generalised form of this model. Deng et al. in \cite{Deng1} studied analytical properties of the generalised form of this model with Poisson-driven jump and revealed weak convergence of EM method.
\par
While the SDE \eqref{eq:1} enjoys significant patronage of both market participants and practitioners, it may also not be well specified to adequately explain interest rate dynamics in response to joint effects of extreme volatility and jump behaviour or information flows as observed empirically from most financial markets. This motivates the need to modify this model to help explain adequately these empirical phenomena more collectively. In modelling context, it is worthwhile to extend SDE \eqref{eq:1} to incorporate delayed volatility function and Poisson-driven jump described by
\begin{equation}\label{eq:2}
  dx(t)=(\alpha_{-1}x(t^-)^{-1}-\alpha_{0}+\alpha_{1}x(t^-)-\alpha_{2}x(t^-)^{\rho})dt+\varphi(x((t-\tau)^-))x(t^-)^{\theta}dB(t)+\alpha_3x(t^-)dN(t)
\end{equation}
on $t\ge-\tau$ with initial data $x(t)=\xi(t)$ for $t\in[-\tau,0]$. Here $x(t^-)=\lim_{s\rightarrow t^-}x(s)$, $x((t-\tau)^-)$ denotes delay in $x(t^-)$, $\varphi(\cdot)$ depends on $x((t-\tau)^-)$ with $\tau>0$. Moreover, $\alpha_{-1},\alpha_{0},\alpha_{1}, \alpha_{2}, \alpha_{3}>0$, $\rho, \theta >1$, $B(t)$ is a scalar Brownian motion and $N(t)$ is a scalar Poisson process independent of $B(t)$ with a scalar compensated Poisson process defined by $\widetilde{N}(t)=N(t)-\lambda t$, where $\lambda$ is the jump intensity.
\par
The \textup{SDDE}  \eqref{eq:2} is characterised by two distinguished features. The delayed volatility function may explain volatility 'smiles' and 'skews' which are common in most financial markets. On the other hand, the Poisson-driven jump may account for responses of interest rates to discontinuous random effects generated in connection with unexpected catastrophic news or lack of information.
\par
It is worth observing that the \textup{SDDE}  \eqref{eq:2} is not analytically tractable and so there is a need to employ an efficient numerical scheme to estimate the exact solution. We cannot in this case employ classical explicit EM method which requires coefficients to be of linear growth (e.g., see \cite{maobook}). Meanwhile, the truncated EM scheme recently developed in \cite{mao4} serves as a useful explicit numerical tool for strong convergent approximation of SDEs with superlinear coefficients. In this work, we aim at investigating the $L^p$ $(\text{where } p\ge 2)$ finite time strong convergence of the truncated EM solutions of system of \textup{SDDE} \eqref{eq:2} under the local Lipschitz condition plus the Khasminskii-type condition. Essentially, this work extends results in \cite{emma} to cope with random jumps.
\par
The rest of the paper is organised as follows: In section 2, we  will study the existence of a unique global solution to \textup{SDDE}  \eqref{eq:2} and show that the solution will always be positive. We will also establish  moment bounds of the exact solution in section 2. In section 3, we will present the truncated EM approximation scheme for \textup{SDDE}  \eqref{eq:2}. Section 4 will be entirely devoted to explore numerical properties of the truncated EM scheme including $L^p$ finite time strong convergence of the truncated EM approximate solutions to the exact solution. In section 5, we will perform some numerical illustrations to support the findings. Finally,  we will apply the strong convergence result within a Monte Carlo framework to value some debt and derivative instruments in section 6.
\section{Analytical properties}
In this section, we establish existence of uniqueness and moment bounds of the exact solution to \textup{SDDE}  \eqref{eq:2}. In sequel, let introduce the following mathematical notations and settings.
\subsection{\textbf{Mathematical preliminaries}}
\noindent Throughout this paper unless otherwise specified, we let $\{\Omega,\mathcal{F}, \{\mathcal{F}_t\}_{t\geq 0}, \mathbb{P} \}$ be a complete probability space with filtration $\{\mathcal{F}_t\}_{t\geq 0}$ satisfying the usual conditions (i.e, it is increasing and right continuous while $\mathcal{F}_0$ contains all $\mathbb{P}$-null sets), and let $\mathbb{E}$ denote the expectation corresponding to $\mathbb{P}$. Let $B(t), t\geq 0$, be a scalar Brownian motion defined on the above probability space. Let $N(t)$ be a scalar Poisson process independent of $B(t)$ with compensated Poisson process $\widetilde{N}(t)=N(t)-\lambda t$, where $\lambda$ is the jump intensity, also defined on the above probability space. If $x, y$ are real numbers, then $x\vee y$ denotes the maximum of $x \text{ and } y$, and $x\wedge y$ denotes the minimum of $x \text{ and } y$. Let $\mathbb{R}=(-\infty,\infty)$ and $\mathbb{R}_+=(0,\infty)$. For $\tau >0$, let $C([-\tau,0];\mathbb{R}_+)$ denote the space of all continuous functions $\xi: [-\tau,0]\rightarrow \mathbb{R}_+$ with the norm $\|\xi\|=\sup_{-\tau\leq t\leq 0}\xi(t)$. For an empty set $\emptyset$, we set $\text{inf }\emptyset=\infty$. For a set $A$, we denote its indication function by $1_A$. Let the following dynamics
\begin{equation}\label{eq:3}
dx(t)=f(x(t^-))dt+\varphi(x((t-\tau)^-))g(x(t^-))dB(t)+h(x(t^-))dN(t),
\end{equation}
$x(t)=\xi(t)$, on $ t\in [-\tau,\infty)$, denote equation of \textup{SDDE}  \eqref{eq:2} such that $f(x)=\alpha_{-1}x^{-1}-\alpha_{0}+\alpha_{1}x-\alpha_{2}x^{\rho}$, $g(x)=x^{\theta}$ and $h(x)=\alpha_{3}x$, $\forall x\in \mathbb{R}_+$, with $\varphi(\cdot)$ defined in $C(\mathbb{R}_+;\mathbb{R}_+)$. Let $C^{2,1}(\mathbb{R}\times \mathbb{R}_+;\mathbb{R})$ be the family of all real-valued functions $Z(x,t)$ defined on $\mathbb{R}\times \mathbb{R}_+$ such that $Z(x,t)$ is twice continuously differentiable in $x$ and once in $t$.  For each $Z\in C^{2,1}(\mathbb{R}\times \mathbb{R}_+;\mathbb{R})$, define the jump-diffusion operator $LZ:\mathbb{R}\times \mathbb{R}\times \mathbb{R}_+\rightarrow \mathbb{R}$ by
\begin{equation}\label{eq:4}
LZ(x,y,t)=\ell(x,y,t)+\lambda(Z(x+h(x),t)-Z(x,t)),
\end{equation}
for SDDE \eqref{eq:3} associated with the $C^{2,1}$-function $Z$, where
\begin{equation}\label{eq:5}
\ell(x,y,t)=Z_t(x,t)+Z_x(x,t)f(x)+\frac{1}{2}Z_{xx}(x,t)\varphi(y)^2g(x)^2,
\end{equation}
$\ell Z:\mathbb{R}\times \mathbb{R}\times \mathbb{R}_+\rightarrow \mathbb{R}$, is the diffusion operator, $Z_t(x,t)$ and $Z_x(x,t)$  are first-order partial derivatives with respect to $t$ and $x$ respectively, and $Z_{xx}(x,t)$, a second-order partial derivative with respect to $x$. With the jump-diffusion operator defined, the It\^{o} formula then yields
\begin{align}\label{eq:7}
dZ(x(t),t)&=LZ(x(t^-),x((t-\tau)^-),t)dt+\varphi(x((t-\tau)^-))Z_x(x(t^-),t)g(x(t^-))dB(t)\\
&+(Z(x(t^-)+h(x(t^-)),t)-Z(x(t^-),t))d\widetilde{N}(t)\nonumber
\end{align}
almost surely. We refer the reader, for instance, to \cite{StoPoisson} for detailed coverage of \eqref{eq:7}. 
\subsection{\textbf{Existence of nonnegative solution}}
\noindent Before we show existence of nonnegative solution to \textup{SDDE}  \eqref{eq:3}, we are required to assume the volatility function $\varphi(\cdot)$ is locally Lipschitz continuous and bounded (see, e.g.,\cite{mao3} for detailed accounts of these conditions). The following conditions are thus sufficient to establish existence of a unique positive global or nonexplosive solution to \textup{SDDE} \eqref{eq:3}.
\begin{assumption}\label{eq:a1}
The volatility function $\varphi:\mathbb{R_+}\rightarrow \mathbb{R_+}$ of \textup{SDDE}  \eqref{eq:3} is Borel-measurable and bounded by a positive constant $\sigma$, i.e.
\begin{equation}\label{eq:8}
  \varphi(y)\leq \sigma,
\end{equation}
$\forall y\in \mathbb{R}_+ $.
\end{assumption}
\begin{assumption}\label{eq:a5}
For any $R>0$, there exists a constant $L_R> 0$ such that the volatility function $\varphi(\cdot)$ of \textup{SDDE}  \eqref{eq:3} satisfies
\begin{equation}\label{eq:20}
|\varphi(y)-\varphi(\bar{y})| \leq L_R |y-\bar{y}|
\end{equation}
$\forall y,\bar{y}\in [1/\mathbb{R},\mathbb{R}]$.
\end{assumption}
\begin{assumption} \label{eq:a2}
The parameters of \textup{SDDE}  \eqref{eq:3} satisfy
\begin{equation}\label{eq:9}
1+ \rho> 2\theta.
\end{equation}
\end{assumption}
The following theorem reveals the \textup{SDDE}  \eqref{eq:3} admits a pathwise-unique positive global solution $x(t)$ on $t\in [-\tau,\infty)$. Since \textup{SDDE}  \eqref{eq:3} describes interest rate dynamics, the solution will always remain nonnegative almost surely.
\begin{theorem}\label{eq:thrm1}
Let Assumptions \ref{eq:a1} and \ref{eq:a2} hold. Then for any given initial data
\begin{equation}\label{eq:10}
\{ x(t): -\tau\leq t \leq 0\}=\xi(t) \in C([-\tau,0]):\mathbb{R}_+),
\end{equation}
there exists a unique global solution $x(t)$ to \textup{SDDE}  \eqref{eq:3} on $t\in [-\tau,\infty)$ and $x(t)>0$ $a.s$.
\end{theorem}
\begin{proof} Since the coefficient terms of \textup{SDDE}  \eqref{eq:3} are locally Lipschitz continuous in $[-\tau,\infty)$, then there exists a unique positive maximal local solution $x(t)\in [-\tau,\tau_{e})$ for any given initial data \eqref{eq:10}, where $\tau_{e}$ is the explosion time. Let  $n_0>0$ be sufficiently large such that
\begin{equation*}\label{eq:11}
  \frac{1}{n_0}<\underset{-\tau\leq t\leq0}{\min}|\xi(t)|\leq\underset{-\tau\leq t\leq 0}{\max}|\xi(t)|<n_0.
\end{equation*}
For each integer $n\geq n_0$, define the stopping time
\begin{equation}\label{eq:12}
\tau_n=\inf\{ t\in [0,\tau_{e}):x(t)\not\in(1/n,n)\}.
\end{equation}
Obviously, $\tau_n$ is increasing  as $n\rightarrow \infty$. Set $\tau_{\infty}=\underset{n\rightarrow \infty}\lim \tau_n$, whence $\tau_{\infty}\leq \tau_e$ almost surely. In other words, we need to show that $ \tau_{\infty}=\infty$ almost surely to complete the proof. For any $\beta\in(0,1)$, define a $C^2$-function $Z:\mathbb{R_+}\rightarrow \mathbb{R_+}$ by
\begin{equation}\label{eq:13}
Z(x)=x^{\beta}-1-\beta\text{log}(x).
\end{equation}
Clearly $Z(x)\rightarrow \infty $ as $x\rightarrow \infty$ or $x\rightarrow 0$. By Assumption \ref{eq:a1}, we get from the operator in \eqref{eq:4} that
\begin{align*}
 LZ(x,y)&\le \ell Z(x,y)+\lambda\Big((x+\alpha_3x)^{\beta}-1-\beta\log(x+\alpha_3x)-(x^{\beta}-1-\beta\log(x))\Big)\\
  &=\ell Z(x,y)+\lambda\Big(((x+\alpha_3x)^{\beta}-x^{\beta})-\beta\log(x(+\alpha_3)/x)\Big)\\
 &= \ell Z(x,y)+\lambda((1+\alpha_3)^{\beta}-1)x^{\beta}-\lambda \beta\log(1+\alpha_3),
\end{align*}
where
\begin{align*}
\ell(x,y)&=\beta(x^{\beta-1}-x^{-1})(\alpha_{-1}x^{-1}-\alpha_0+\alpha_1x-\alpha_2x^{\rho})+\frac{1}{2}(\beta(\beta-1)x^{\beta-2}
+\beta x^{-2})\varphi(y)^2 x^{2\theta}\\
&\le  \alpha_{-1}\beta x^{\beta-2}-\alpha_0\beta x^{\beta-1}+\alpha_1\beta x^{\beta}-\alpha_2\beta x^{\rho+\beta-1}-\alpha_{-1}\beta x^{-2}+\alpha_0\beta x^{-1}\\
&-\alpha_1\beta+\alpha_2\beta x^{\rho-1}+\frac{\sigma^2}{2}\beta(\beta-1)x^{\beta+2\theta-2}+\frac{\sigma^2}{2}\beta x^{2\theta-2}.
\end{align*}
Since $\beta \in(0,1)$ and by Assumption \ref{eq:a2}, we note $-\alpha_{-1}\beta x^{-2}$ leads and tends to $-\infty$ for small $x$ and for large $x$, $-\alpha_2\beta x^{\rho+\beta-1}$ leads and also tends to $-\infty$. Hence there exists a constant $K_0$ such that
\begin{equation}\label{eq:14}
LZ(x,y)\leq K_0.
\end{equation}
So for $t_1\in [0,\tau]$, we derive from the It\^{o} formula 
\begin{align*}
\mathbb{E}[Z(x(\tau_n\wedge t_1))]\leq Z(\xi(0))+\int_0^{\tau_n\wedge t_1}K_0dt,
\end{align*}
$\forall n\geq n_0$. It then follows that
\begin{equation*}
\mathbb{P}(\tau_n\leq \tau)\leq\frac{Z(\xi(0))+K_0\tau}{Z(1/n)\wedge Z(n)}.
\end{equation*}
As $n\rightarrow \infty$, $ \mathbb{P}(\tau_n\leq \tau)\rightarrow 0 $. This implies $ \tau_{\infty}>\tau$ a.s. Also for $t_1\in [0,2\tau]$, the It\^{o} formula yields
\begin{align*}
\mathbb{E}[Z(x(\tau_n\wedge t_1))]\leq Z(\xi(0))+\int_0^{\tau_n\wedge t_1}K_0dt,
\end{align*}
$\forall n\geq n_0$ and consequently,
\begin{equation*}\label{eq:15}
\mathbb{P}(\tau_n\leq 2\tau)\leq\frac{Z(\xi(0))+2K_0\tau}{Z(1/n)\wedge Z(n)}.
\end{equation*}
As $n\rightarrow \infty$, we get $ \tau_{\infty}>2\tau$ a.s. Repeating this procedure for $t_1\in[0, \infty]$, we obtain $\mathbb{P}(\tau_{\infty}\leq \infty)\rightarrow 0$ by letting $n\rightarrow \infty$. This means $ \tau_{\infty}=\infty $ a.s and hence $ \tau_{e}=\infty$ a.s. The proof is now complete.
\end{proof}
\subsection{\textbf{Moment bound}}
\noindent The following lemmas show moments of the exact solution to \textup{SDDE}  \eqref{eq:3} are upper bounded.
\begin{lemma}\label{eq:l1}
Let Assumptions \ref{eq:a1} and \ref{eq:a2} hold. Then for any $p\geq 2$, there exists a constant $\rho_1$ such that the solution of \textup{SDDE}  \eqref{eq:3} satisfies
\begin{equation}\label{eq:16}
\sup_{0\leq t<\infty}\Big(\mathbb{E}|x(t)|^p\Big)\leq \rho_1.
\end{equation}
\end{lemma}
\begin{proof}
Define the stopping time for every sufficiently large integer $n$ by
\begin{equation}\label{stopt}
\tau_n=\inf \{ t\geq 0:x(t)\not\in(1/n,n)\}.
\end{equation}
Define a function $Z\in C^{2,1}(\mathbb{R}_+\times \mathbb{R}_+;\mathbb{R}_+)$ by $Z(x,t)=e^tx^p$ . By Assumption \ref{eq:a1}, the jump-diffusion operator in \eqref{eq:4} gives us
\begin{align*}
 LZ(x,y,t) &\le  \ell Z(x,y,t)+\lambda [e^t(x+\alpha_{3}x)^p-e^tx^p]\\
  & = \ell Z(x,y,t)+\lambda e^tx^p[(1+\alpha_{3})^p-1],
\end{align*}
where
\begin{align*}
\ell Z(x,y,t)&=e^tx^p+pe^tx^{p-1}(\alpha_{-1}x^{-1}-\alpha_{0}+\alpha_{1}x-\alpha_{2}x^{\rho})+\frac{1}{2}p(p-1)e^tx^{p-2}\varphi^2(y)x^{2\theta}\\
&\le e^t[x^p+\alpha_{-1}px^{p-2}-\alpha_{0}px^{p-1}+\alpha_{1}px^{p}-\alpha_{2}px^{\rho+p-1}+\frac{p(p-1)}{2}\sigma^2x^{2\theta+p-2})].
\end{align*}
By Assumption \ref{eq:a2}, $-p\alpha_{2}x^{\rho+p-1}$ dominates and tends to $-\infty$ for large $x$. Hence we can find a constant $K_1$ such that
\begin{equation*}
  LZ(x,y,t)\leq K_1e^t.
\end{equation*}
The It\^{o} formula gives us
\begin{equation*}
\mathbb{E}[e^{t\wedge \tau_n}|x(t\wedge \tau_n)|^p]\leq |\xi(0)|^p+K_1e^t.
\end{equation*}
Applying  the Fatou lemma and letting $n\rightarrow \infty$ yields
\begin{equation*}
\mathbb{E}|x(t)|^p<e^{-t}|\xi(0)|^p+K_1
\end{equation*}
and consequently,
\begin{equation*}
\underset{0\leq t< \infty}{\sup}(\mathbb{E}|x(t)|^p)\leq \rho_1.
\end{equation*}
as the required assertion in \eqref{eq:16}.
\end{proof}
\begin{lemma}\label{eq:l1*}
Let Assumptions \ref{eq:a1} and \ref{eq:a2} hold. For any $p >2 \vee (\rho-1)$, there exists a constant $\rho_2$ such that the solution of \textup{SDDE}  \eqref{eq:3} satisfies
\begin{equation} \label{eq:17}
\sup_{0\leq t<\infty}\Big(\mathbb{E}|\frac{1}{x(t)}|^p\Big)\leq \rho_2.
\end{equation}
\end{lemma}
\begin{proof}
Let $\tau_n$ be the same as in \eqref{stopt}. By applying \eqref{eq:4} to $Z(x,t)=e^t/x^p$, we compute
\begin{align*}
 LZ(x,y,t) &\le \ell Z(x,y,t)+\lambda[e^t(x+\alpha_3x)^{-p}-e^tx^{-p}]\\
  &= \ell Z(x,y,t)+\lambda e^tx^{-p}[(1+\alpha_{3})^{-p}-1],
\end{align*}
where  Assumption \ref{eq:a1} has been used and here, we have
\begin{align*}
\ell Z(x,y,t)&= e^tx^{-p}-pe^tx^{-(p+1)}(\alpha_{-1}x^{-1}-\alpha_{0}+\alpha_{1}x-\alpha_{2}x^{\rho})+\frac{1}{2}p(p+1)e^tx^{-(p+2)}\varphi(y)^2x^{2\theta}\\
&\leq e^t[x^{-p}-\alpha_{-1}px^{-(p+2)}+\alpha_{0}px^{-(p+1)}-\alpha_{1}px^{-p}+\alpha_{2}x^{\rho-p-1}-\frac{p(p+1)}{2}\sigma^2x^{2\theta-p-2})].
\end{align*}
By Assumption \ref{eq:a2} and noting that $p >2 \vee (\rho-1)$, we observe $-\alpha_{-1}px^{-(p+2)}$ leads and tends to $-\infty$ for small $x$ and for large $x$, $p\alpha_{2}x^{\rho-p-1}$ dominates and tends to 0. Hence there exists a constant $K_2$ such that
\begin{equation*}
  LZ(x,y,t,)\leq K_2e^t.
\end{equation*}
We can now use the It\^{o} formula, apply Fatou lemma and let $n\rightarrow \infty$ to arrive at \begin{equation*}
\mathbb{E}|x(t)|^{-p}<e^{-t}|\xi(0)|^{-p}+K_2
\end{equation*}
and consequently the required assertion in \eqref{eq:17}.
\end{proof}
\section{The truncated EM method}
\noindent In this section, we present the truncated EM scheme for numerical approximation of \textup{SDDE}  \eqref{eq:3}. Meanwhile,  we need the following assumption on the initial data which will be used later.
\begin{assumption}\label{eq:a3}
There is a pair of constant $K_3>0$ and $\gamma\in (0,1]$ such that for all $ -\tau\leq s \leq t \leq 0$, the initial data $\xi$ satisfies
\begin{equation}\label{eq:18}
  |\xi(t)-\xi(s)|\leq K_3|t-s|^\gamma.
\end{equation}
\end{assumption}
\noindent In the sequel, we also need these lemmas below.
\begin{lemma}\label{eq:l2}
For any $R>0$, there exists a constant $K_R > 0$ such that the coefficient terms $f$, $g$ and $h$ of \textup{SDDE}  \eqref{eq:3} satisfy
\begin{equation}\label{eq:21}
|f(x)-f(\bar{x})|\vee |g(x)-g(\bar{x})|\vee |h(x)-h(\bar{x})| \leq K_R|x-\bar{x}|,
\end{equation}
$\forall x,\bar{x} \in [1/\mathbb{R},\mathbb{R}]$.
\end{lemma}
\begin{lemma}\label{eq:l3}
Let Assumptions \ref{eq:a1} and \ref{eq:a2} hold. Then for any $p\geq 2$, there exists $K_4>0$ such that the drift and diffusion terms of \textup{SDDE}  \eqref{eq:3} satisfy
\begin{equation}\label{eq:22}
xf(x)+\frac{p-1}{2}|\varphi(y)g(x)|^2\leq K_4(1+|x|^2),
\end{equation}
$\forall x,y\in \mathbb{R}_+$, where $K_4$ is a constant (see  \cite{emma} for the proof).
\end{lemma}
\subsection{\textbf{Numerical approximation}}
\noindent Before we proceed, let extend the volatility function $\varphi(\cdot)$ and the jump term $h(\cdot)$ from $\mathbb{R}_+$ to $\mathbb{R}$ by setting $\varphi(y)=\varphi(0)$ and $h(x)=0$ for $x<0$. Apparently, Theorem \ref{eq:thrm1}  as well as conditions \eqref{eq:8}, \eqref{eq:20}, \eqref{eq:21} and \eqref{eq:22} are well maintained. Moreover, we need not truncate the jump term since it is of linear growth. To define the truncated EM scheme for \textup{SDDE}  \eqref{eq:3}, we first choose a strictly increasing continuous function $\mu: \mathbb{R}_+ \rightarrow \mathbb{R}_+$ such that $\mu(r)\rightarrow \infty$ as $r\rightarrow \infty$ and
\begin{equation}\label{eq:23}
\sup_{1/r \leq x\leq r}(|f(x)|\vee g(x))\leq \mu(r), \quad \forall r> 1.
\end{equation}
Denote by $\mu^{-1}$ the inverse function of $\mu$. We define a strictly decreasing function
$\pi:(0,1)\rightarrow \mathbb{R}_+$ such that
\begin{equation}\label{eq:24}
\quad \lim_{\Delta \rightarrow 0}\pi(\Delta)=\infty \text{  and  }\Delta^{1/4}\pi(\Delta)\leq 1, \quad \forall \Delta \in (0,1].
\end{equation}
Find $\Delta^*\in (0,1)$ such that $\mu^{-1}(\pi(\Delta^*))>1$ and $f(x)>0$ for $0<x<\Delta^*$. For a given step size $\Delta \in (0,\Delta^*)$, let us define the truncated functions
\begin{equation*}
f_{\Delta}(x)=f\Big(1/\mu^{-1}(\pi(\Delta))\vee (x\wedge \mu^{-1}(\pi(\Delta)))\Big), \quad \forall x \in \mathbb{R}
\end{equation*}
and
\begin{equation*}
g_{\Delta}(x)=
\begin{cases}
  g\Big(x\wedge \mu^{-1}(\pi(\Delta))\Big), & \mbox{if $x\geq 0$ }\\
  0,                             & \mbox{if $x< 0$}.
\end{cases}
\end{equation*}
So for $x\in [1/\mu^{-1}(\pi(\Delta)),\mu^{-1}(\pi(\Delta))]$, we  have
\begin{align*}
|f_{\Delta}(x)|&=|f(x)|\leq \underset{1/\mu^{-1}(\pi(\Delta))\leq w\leq\mu^{-1}(\pi(\Delta))}{\max |f(w)|}\\
             &\leq \mu(\mu^{-1}(\pi(\Delta)))= \pi(\Delta)
\end{align*}
and
\begin{align*}
g_{\Delta}(x)\leq \mu(\mu^{-1}(\pi(\Delta)))=\pi(\Delta).
\end{align*}
We easily observe that
\begin{equation}\label{eq:25}
|f_{\Delta}(x)|\vee g_{\Delta}(x)\leq \pi(\Delta), \quad \forall x \in \mathbb{R}.
\end{equation}\noindent
That is, both truncated functions $f_{\Delta}$ and $g_{\Delta}$ are bounded although both $f$ and $g$ may not. The following lemma shows $f_{\Delta}$ and $g_{\Delta}$ maintain \eqref{eq:22} nicely.
\begin{lemma}\label{eq:l4}
Let Assumptions \ref{eq:a1} and \ref{eq:a2} hold. Then, for all $\Delta \in (0,\Delta^*)$ and $p\geq 2$, the truncated functions satisfy
\begin{equation}\label{eq:26}
xf_{\Delta}(x)+\frac{p-1}{2}|\varphi(y)g_{\Delta}(x)|^2\leq K_5(1+|x|^2)
\end{equation}
$\forall x,y\in \mathbb{R}$, where $K_5$ is a constant independent of $\Delta$ (see  \cite{emma} for the proof).
\end{lemma}
\noindent From now on, let $T>0$ be fixed arbitrarily and the step size $\Delta\in (0,\Delta^*]$ be a fraction of $\tau$.  We define $\Delta=\tau/M$ for some positive integer $M$. Let now form the discrete-time truncated EM approximation of \textup{SDDE} \eqref{eq:3}. Define $t_k=k\Delta$ for $k=-M,-(M-1),..,0,1,2,..$. Set $X_{\Delta}(t_k)= \xi(t_k)$ for $k=-M,-(M-1),..,0$  and then compute
\begin{equation}\label{eq:27}
X_{\Delta}(t_{k+1})=X_{\Delta}(t_k)+f_{\Delta}(X_{\Delta}(t_k))\Delta+\varphi(X_{\Delta}(t_{k-M}))g_{\Delta}(X_{\Delta}(t_{k}))\Delta B_k+h(X_{\Delta}(t_{k}))\Delta N_k
\end{equation}
for $k=0,1,2...,$ where $\Delta B_k=B(t_{k+1})-B(t_k)$ and $\Delta N_k=N(t_{k+1})-N(t_k)$. Let now form two versions of the continuous-time truncated EM solutions. The first one is defined by
\begin{equation}\label{eq:28}
\bar{x}_{\Delta}(t)=\sum_{k=-M}^{\infty}X_{\Delta}(t_k)1_{[t_k,t_{k+1})}(t).
\end{equation}
This is the continuous-time step-process $\bar{x}_{\Delta}(t)$ on $t\in [-\tau, \infty]$, where $1_{[t_k,t_{k+1})}$ is the indicator function on $[t_k,t_{k+1})$. The other one is the continuous-time continuous process $x_{\Delta}(t)$ on $t\ge -\tau$ defined conveniently by setting $x_{\Delta}(t)=\xi(t)$ for $t\in [-\tau,0]$ while for $t\geq 0$
\begin{equation}\label{eq:29}
x_{\Delta}(t)=\xi(0)+\int_0^t f_{\Delta}(\bar{x}_{\Delta}(s^-))ds+\int_0^t \varphi(\bar{x}_{\Delta}((s-\tau)^-))g_{\Delta}(\bar{x}_{\Delta}(s^-))dB(s)+\int_0^t h(\bar{x}_{\Delta}(s^-))dN(s).
\end{equation}
Obviously $x_{\Delta}(t)$ is an It\^{o} process on $t\geq 0$ satisfying It\^{o} differential
\begin{equation}\label{eq:30}
dx_{\Delta}(t)= f_{\Delta}(\bar{x}_{\Delta}(t^-))dt+\varphi(\bar{x}_{\Delta}((t-\tau)^-))g_{\Delta}(\bar{x}_{\Delta}(t^-))dB(t)+h(\bar{x}_{\Delta}(t^-))dN(t).
\end{equation}
\noindent For all $k=-M,-(M-1),..$, it is useful to see that $x_{\Delta}(t_{k})=\bar{x}_{\Delta}(t_k)=X_{\Delta}(t_k)$.
\section{\textbf{Numerical properties}}
In this section, we establish moment bound and finite time strong convergence theory of the truncated EM solutions to \textup{SDDE} \eqref{eq:3}.
\subsection{\textbf{Moment bound}}
To upper bound the moment of the truncated EM solution, let first define 
\begin{equation*}
k(t)=[t/\Delta]\Delta,
\end{equation*}
for any $t\in [0,T]$, where $[t/\Delta]$ denotes the integer part of $t/\Delta$. The following lemma shows $x_{\Delta}(t)$ and $\bar{x}_{\Delta}(t)$ are close to each other in $L^p$.  
\begin{lemma}\label{eq:l5}
Let Assumption \ref{eq:a1} hold. Then for any fixed $\Delta\in (0,\Delta^*]$, we have
\begin{equation}\label{eq:31}
 \mathbb{E}\Big(|x_{\Delta}(t)-\bar{x}_{\Delta}(t)|^{p}\big| \mathcal{F}_{k(t)}\Big)\leq \mathfrak{D}_1\Big(\Delta^{p/2}(\pi(\Delta))^{p}+\Delta\Big)|\bar{x}_{\Delta}(t)|^{p}, \quad  p\in[2,\infty)
\end{equation}
and 
\begin{equation}\label{eq:31*}
 \mathbb{E}\Big(|x_{\Delta}(t)-\bar{x}_{\Delta}(t)|^{p}\big| \mathcal{F}_{k(t)}\Big)\leq \mathfrak{D}_2\Big(\Delta^{p/2}(\pi(\Delta))^{p}\Big)|\bar{x}_{\Delta}(t)|^{p}, \quad  p\in(0,2),
\end{equation}
for all $t\geq 0$, where $\mathfrak{D}_1$ and $\mathfrak{D}_2$ denote positive generic constants which depend only on $p$ and may change between occurrences.
\end{lemma}
\begin{proof}
Fix any $\Delta \in (0,\Delta^*)$ and $t\in [0,T]$. Then for $ p\in[2,\infty)$, we derive 
\begin{align*}
&\mathbb{E}\Big(|x_{\Delta}(t)-\bar{x}_{\Delta}(t)|^{p}\big| \mathcal{F}_{k(t)}\Big)\\
&\le 3^{p-1}\Big(\mathbb{E}\big(| \int_{k(t)}^{t} f_{\Delta}(\bar{x}_{\Delta}(s))ds|^{p}\big| \mathcal{F}_{k(t)}\big)+\mathbb{E}\big(|\int_{k(t)}^t \varphi(\bar{x}_{\Delta}((s-\tau)))g_{\Delta}(\bar{x}_{\Delta}(s))dB(s)|^{p}\big| \mathcal{F}_{k(t)}\big)\\
&+\mathbb{E}\big(|\int_{k(t)}^t h(\bar{x}_{\Delta}(s))dN(s)|^{p}\big| \mathcal{F}_{k(t)}\big) \Big)\\
&\leq 3^{p-1} \Big(\Delta^{p-1}\mathbb{E}( \int_{k(t)}^t |f_{\Delta}(\bar{x}_{\Delta}(s))|^{p} ds\big| \mathcal{F}_{k(t)})+ c(p)\Delta^{(p-2)/2} \mathbb{E}(\int_{k(t)}^{t}|\varphi(\bar{x}_{\Delta}((s-\tau)))g_{\Delta}(\bar{x}_{\Delta}(s))|^{p} ds\big| \mathcal{F}_{k(t)})\\
&+\mathbb{E}(|\int_{k(t)}^t h(\bar{x}_{\Delta}(s))dN(s)|^{p}\big| \mathcal{F}_{k(t)} )\Big)\\
&\leq 3^{p-1}\Big(\Delta^{p-1}\Delta (\pi(\Delta))^{p} +c(p)\Delta^{(p-2)/2}\Delta (\sigma \pi(\Delta))^{p}+\mathbb{E}(|\int_{k(t)}^t h(\bar{x}_{\Delta}(s))dN(s)|^{p}\big|
\mathcal{F}_{k(t)}) \Big),
\end{align*}
where Assumption \ref{eq:a1} and \eqref{eq:25} have been used and $c(p)$ depends on $p$. By the characteristic function's argument (see \cite{HiKlo}), we have
\begin{equation*}
\mathbb{E}|\Delta N_k|^{p} \leq \bar{c}\Delta, \quad \forall \Delta\in (0,\Delta^*),
\end{equation*}
where $\bar{c}$ is a positive constant independent of  $\Delta$. We now obtain 
\begin{align*}
\mathbb{E}(|\int_{k(t)}^t h(\bar{x}_{\Delta}(s))dN(s)|^{p}\big|\mathcal{F}_{k(t)})
&=|h(\bar{x}_{\Delta}(t))|^{p}\mathbb{E}|\Delta N_k|^{p}.
\end{align*}
This implies 
\begin{align*}
\mathbb{E}\Big(|x_{\Delta}(t)-\bar{x}_{\Delta}(t)|^{p}\big| \mathcal{F}_{k(t)}\Big)
&\leq 3^{p-1}\Big(\Delta^{p-1}\Delta (\pi(\Delta))^{p} +c(p)\Delta^{(p-2)/2}\Delta (\sigma \pi(\Delta))^{p}+|h(\bar{x}_{\Delta}(t))|^{p} \mathbb{E}|\Delta N_k|^{p} \Big),
\end{align*}
where $h(\bar{x}_{\Delta}(t))$ is independent of $ N_k$. We now have
\begin{align*}
\mathbb{E}\Big(|x_{\Delta}(t)-\bar{x}_{\Delta}(t)|^{p}\big| \mathcal{F}_{k(t)}\Big)
&\leq 3^{p-1}\Big((1\vee c(p)\sigma^{p})\Delta^{p/2}(\pi(\Delta))^{p}+\bar{c} \alpha_3^{p}|\bar{x}_{\Delta}(t^-)|^{p}\Delta \Big)\\
&\leq 3^{p-1}(1\vee c(p)\sigma^{p}\vee \bar{c}\alpha_3^{p})\Big(\Delta^{p/2}(\pi(\Delta))^{p}+|\bar{x}_{\Delta}(t)|^{p}\Delta \Big)\\
&\leq \mathfrak{D}_1\Big(\Delta^{p/2}(\pi(\Delta))^{p}+\Delta \Big)|\bar{x}_{\Delta}(t)|^{p},
\end{align*}
which is \eqref{eq:31}, where $\mathfrak{D}_1=3^{p-1}[(1\vee c(p)\sigma^{p}) \vee \bar{c}\alpha_3^{p}]$. For $p\in (0,2)$, the Jensen inequality yields
\begin{align*}
\mathbb{E}\Big(|x_{\Delta}(t)-\bar{x}_{\Delta}(t)|^{p}\big| \mathcal{F}_{k(t)}\Big)
&\le \Big\{ \mathbb{E}\Big(|x_{\Delta}(t)-\bar{x}_{\Delta}(t)|^{2}\big| \mathcal{F}_{k(t)}\Big)\Big\}^{p/2}\\
&\le \Big\{\mathfrak{D}_1\Big(\Delta(\pi(\Delta))^{2}+\Delta \Big)|\bar{x}_{\Delta}(t)|^{p}\Big\}^{p/2}\\
&\le 2^{p/2-1}\mathfrak{D}_1^{p/2}\Big(\Delta^{p/2}(\pi(\Delta))^{p}+\Delta^{p/2} \Big)(|\bar{x}_{\Delta}(t)|^{p})^{p/2}\\
&\le \mathfrak{D}_2\Big(\Delta^{p/2}(\pi(\Delta))^{p}\Big)|\bar{x}_{\Delta}(t)|^{p},
\end{align*}
which is the required assertion in \eqref{eq:31*}, where $\mathfrak{D}_2=2^{p/2}\mathfrak{D}_1^{p/2}$. The proof is thus complete.
\end{proof}
 We can now upper bound the moment of the truncated EM solution as follows.
\begin{lemma}\label{eq:l6}
Let Assumptions \ref{eq:a1} and \ref{eq:a2} hold. Then for any $ p\ge 3$
\begin{equation}\label{eq:32}
  \sup_{0\leq \Delta \leq \Delta^*} \sup_{0\leq t\leq T}(\mathbb{E}|x_{\Delta}(t)|^p)\leq \rho_3, \quad \forall T> 0,
\end{equation}
where $\rho_3:=\rho_3(T, p, K, \xi)$ and may change between occurrences.
\end{lemma}
\begin{proof}
Fix any $\Delta \in (0,\Delta^*)$ and $T\geq 0$. For $t\in[0,T]$, we obtain from  \eqref{eq:4}, \eqref{eq:23} and Lemma \ref{eq:l4} 
\begin{align*}
\mathbb{E}|x_{\Delta}(t)|^p-|\xi(0)|^p &\leq\mathbb{E}\int_{0}^{t}p|x_{\Delta}(s^-)|^{p-2}\Big(\bar{x}_{\Delta}(s^-)f_{\Delta}(\bar{x}_{\Delta}(s^-))+\frac{p-1}{2}|\varphi(\bar{x}_{\Delta}((s-\tau)^-))g_{\Delta}(\bar{x}_{\Delta}(s^-))|^2\Big)ds\\
&+\mathbb{E}\int_{0}^{t}p|x_{\Delta}(s^-)|^{p-2}(x_{\Delta}(s^-)-\bar{x}_{\Delta}(s^-))f_{\Delta}(\bar{x}_{\Delta}(s^-))ds\\
&+\lambda\mathbb{E}\Big(\int_{0}^{t}|x_{\Delta}(s^-)+h(\bar{x}_{\Delta}(s^-))|^p-|x_{\Delta}(s^-)|^p\Big)ds\\
&\leq H_{11}+H_{12}+H_{13},
\end{align*}
where
\begin{align*}
H_{11}&=\mathbb{E}\int_{0}^{t}K_5p|x_{\Delta}(s^-)|^{p-2}(1+|\bar{x}_{\Delta}(s^-)|^2)ds\\
H_{12}&=\mathbb{E}\int_{0}^{t}p|x_{\Delta}(s^-)|^{p-2}\Big(x_{\Delta}(s^-)-\bar{x}_{\Delta}(s^-)\Big)f_{\Delta}(\bar{x}_{\Delta}(s^-))ds\\
H_{13}&=\lambda\mathbb{E}\Big(\int_{0}^{t}|x_{\Delta}(s^-)+h(\bar{x}_{\Delta}(s^-))|^p-|x_{\Delta}(s^-)|^p\Big)ds.
\end{align*}
Applying the Young inequality, we obtain
\begin{align*}
H_{11}&= K_5p\mathbb{E}\int_{0}^{t}|x_{\Delta}(s^-)|^{p-2}(1+|\bar{x}_{\Delta}(s^-)|^2)ds\\
   &\leq K_5p\int_{0}^{t}\Big(\frac{(p-2)}{p}\mathbb{E}|x_{\Delta}(s^-)|^{p}+\frac{2}{p}\mathbb{E}(1+|\bar{x}_{\Delta}(s^-)|)^{p}\Big)ds\\
 &\leq K_5\int_{0}^{t}\Big((p-2)\mathbb{E}|x_{\Delta}(s^-)|^{p}+2^p(1 +\mathbb{E}|\bar{x}_{\Delta}(s^-)|^{p})\Big)ds\\
&\leq c_1\int_{0}^{t}(1+\mathbb{E}|x_{\Delta}(s)|^{p}+\mathbb{E}|\bar{x}_{\Delta}(s)|^{p})ds,
\end{align*}
where $c_1=K_5[(p-2)\vee 2^p]$. For $s\in[0,t]$, we note from the triangle inequality 
\begin{equation*}
|x_{\Delta}(s^-)|\le |x_{\Delta}(s^-)-\bar{x}_{\Delta}(s^-)|+|\bar{x}_{\Delta}(s^-)|.
\end{equation*}
This implies for $p\ge 3$, we obtain
\begin{align*}
H_{12}&\le p\mathbb{E}\int_{0}^{t}\Big(|x_{\Delta}(s^-)-\bar{x}_{\Delta}(s^-)|+|\bar{x}_{\Delta}(s^-)|\Big)^{p-2}|x_{\Delta}(s^-)-\bar{x}_{\Delta}(s^-)||f_{\Delta}(\bar{x}_{\Delta}(s^-))|ds\\
&\le 2^{(p-3)}p\mathbb{E}\int_{0}^{t}\Big(|x_{\Delta}(s^-)-\bar{x}_{\Delta}(s^-)|^{p-2}+|\bar{x}_{\Delta}(s^-)|^{p-2}\Big)|x_{\Delta}(s^-)-\bar{x}_{\Delta}(s^-)||f_{\Delta}(\bar{x}_{\Delta}(s^-))|ds\\
&= H_{121}+H_{122},
\end{align*}
where 
\begin{align*}
H_{121}&= 2^{(p-3)}p\mathbb{E}\int_{0}^{t}|\bar{x}_{\Delta}(s^-)|^{p-2}|x_{\Delta}(s^-)-\bar{x}_{\Delta}(s)||f_{\Delta}(\bar{x}_{\Delta}(s^-))|ds\\
H_{122}&= 2^{(p-3)}p\mathbb{E}\int_{0}^{t}|x_{\Delta}(s^-)-\bar{x}_{\Delta}(s^-)|^{p-1}|f_{\Delta}(\bar{x}_{\Delta}(s^-))|ds.
\end{align*}
By Lemma \ref{eq:l5} and \eqref{eq:25}, we now have
\begin{align}\label{eq:h0}
H_{121}&\le 2^{(p-3)}p\int_{0}^{t}\mathbb{E}\Big\{|\bar{x}_{\Delta}(s)|^{p-2}|f_{\Delta}(\bar{x}_{\Delta}(s))|\mathbb{E}\Big( |x_{\Delta}(s)-\bar{x}_{\Delta}(s)|\mathcal{F}_{k(s)})\Big)\Big\}ds\nonumber \\
&\le 2^{(p-3)}p\mathfrak{D}_2(\pi(\Delta))\Delta^{1/2}(\pi(\Delta))\int_{0}^{t}\mathbb{E}\Big\{|\bar{x}_{\Delta}(s)|(|\bar{x}_{\Delta}(s)|^{p-2})\Big\}ds\nonumber \\
&\le 2^{(p-3)}p\mathfrak{D}_2(\pi(\Delta))\Delta^{1/2}(\pi(\Delta))\int_{0}^{t}\mathbb{E}|\bar{x}_{\Delta}(s)|^{p-1}ds\nonumber \\
&\le 2^{(p-3)}p\mathfrak{D}_2(\pi(\Delta))^2\Delta^{1/2}\int_{0}^{t}\Big(\frac{1}{p}+\frac{(p-1)}{p}\mathbb{E}|\bar{x}_{\Delta}(s)|^{p}\Big)ds\nonumber\\
&\le c_2+c_3\int_{0}^{t}\mathbb{E}|\bar{x}_{\Delta}(s)|^{p}ds,
\end{align}
where $c_2=2^{(p-3)}\mathfrak{D}_2T$ and $c_3=2^{(p-3)}\mathfrak{D}_2(p-1)$ , noting that $(\pi(\Delta))\Delta^{1/4}\le 1$ and hence $$[(\pi(\Delta))\Delta^{1/4}]^2\le 1.$$ Also by \eqref{eq:25},  we have 
\begin{align}\label{eq:h1}
H_{122}&\le 2^{(p-3)}p\pi(\Delta) \int_{0}^{t}\mathbb{E}|x_{\Delta}(s)-\bar{x}_{\Delta}(s)|^{p-1}ds.
\end{align}
Do note for $p\ge 3$ and $\bar{w}\in(0,1/4]$, we have $p\bar{w}\le (p-1)/2$ and then 
\begin{equation}\label{eq:h2}
\Delta^{(p-1)/2-\bar{w}p}\le 1.
\end{equation}
So for $p\ge 3$ and $\bar{w}=1/4$, we obtain from \eqref{eq:h1}, Lemma \ref{eq:l5}, \eqref{eq:h2} and the Young's inequality
\begin{align*}
H_{122}&\le 2^{(p-3)}p\mathfrak{D}_1 \Big(\Delta^{(p-1)/2}(\pi(\Delta))^{p-1}(\pi(\Delta))+\Delta(\pi(\Delta))\Big)\int_{0}^{t}\mathbb{E}|\bar{x}_{\Delta}(s)|^{p-1}ds\\
&\le  2^{(p-3)}p\mathfrak{D}_1\Big(\Delta^{(p-1)/2}(\pi(\Delta))^{p}+\Delta(\pi(\Delta))\Big)\int_{0}^{t}\mathbb{E}|\bar{x}_{\Delta}(s)|^{p-1}ds\\
&\le 2^{(p-3)}p\mathfrak{D}_1\Big(\Delta^{(p-2)/4}+\Delta(\pi(\Delta))\Big)\int_{0}^{t}\mathbb{E}|\bar{x}_{\Delta}(s)|^{p-1}ds\\
&\le 2^{(p-2)}p\mathfrak{D}_1\int_{0}^{t}\Big(\frac{1}{p}+\frac{(p-1)}{p}\mathbb{E}|\bar{x}_{\Delta}(s)|^{p}\Big)ds\\
&\le c_4+ c_5\int_{0}^{t}\mathbb{E}|\bar{x}_{\Delta}(s)|^{p}ds,
\end{align*}
where $c_4=2^{(p-2)}\mathfrak{D}_1T$ and $c_5=2^{(p-2)}\mathfrak{D}_1(p-1)$. We now combine $H_{121}$ and $H_{122}$ to have
\begin{align*}
H_{12}&\le c_2+c_4+(c_3 +c_5)\int_{0}^{t}\mathbb{E}|\bar{x}_{\Delta}(s)|^{p}ds\\
&\le c_6+c_7\int_{0}^{t}\mathbb{E}|\bar{x}_{\Delta}(s)|^{p}ds,
\end{align*}
where $c_6=c_2+c_4$ and $c_7=c_3+c_5$.  Also we estimate $H_{13}$ as
\begin{align*}
H_{13}&=\lambda\mathbb{E}\Big(\int_{0}^{t}|x_{\Delta}(s^-)+h(\bar{x}_{\Delta}(s^-))|^p-|x_{\Delta}(s^-)|^p\Big)ds\\
   &\leq \lambda\mathbb{E}\Big(\int_{0}^{t}2^{p-1}|x_{\Delta}(s^-)|^p+2^{p-1}|h(\bar{x}_{\Delta}(s^-))|^p-|x_{\Delta}(s^-)|^p\Big)ds\\
   &\leq \lambda\mathbb{E}\Big(\int_{0}^{t}(2^{p-1}-1)|x_{\Delta}(s^-)|^p+2^{p-1}\alpha_3^p|\bar{x}_{\Delta}(s^-)|^p\Big)ds\\
   &\leq c_8\int_{0}^{t}(\mathbb{E}|x_{\Delta}(s)|^p+\mathbb{E}|\bar{x}_{\Delta}(s)|^p)ds,
\end{align*}
where $c_8=\lambda[(2^{p-1}-1)\vee 2^{p-1}\alpha_3^p]$. Combining $H_{11}$, $H_{12}$ and $H_{13}$, we have
\begin{align*}
\mathbb{E}|x_{\Delta}(t)|^p &\le |\xi(0)|^p+(c_1T+c_6)+\int_{0}^{t}\Big((c_1+c_8)\mathbb{E}|x_{\Delta}(s)|^{p}+(c_1+c_7+c_8)\mathbb{E}|\bar{x}_{\Delta}(s)|^{p}\Big)ds\\
&\leq c_9+2c_{10}\int_{0}^{t}\sup_{0\leq u \leq s}\Big(\mathbb{E}|x_{\Delta}(u)|^p \Big)ds,
\end{align*}
where $c_9=|\xi(0)|^p+c_1T+c_6$ and $c_{10}=(c_1+c_8)\vee(c_1+c_7+c_8)$. As this holds for any $t\in [0,T]$, we then have
\begin{equation*}
  \sup_{0\leq u \leq t}(\mathbb{E}|x_{\Delta}(u)|^p)\leq  c_9+2c_{10}\int_{0}^{t}\sup_{0\leq u \leq s}\Big(\mathbb{E}|x_{\Delta}(u)|^p \Big)ds.
\end{equation*}
The Gronwall inequality yields
\begin{equation*}
  \sup_{0\leq u \leq T}(\mathbb{E}|x_{\Delta}(u)|^p)\leq \rho_3
\end{equation*}
as the required assertion, where $\rho_3=c_9e^{2c_{10}T}$ is independent of $\Delta$.
\end{proof}
\subsection{\textbf{Finite time strong convergence}}
\noindent We can now establish finite time strong convergence theory for the truncated EM solutions to \textup{SDDE}  \eqref{eq:3}. Before that, let first establish the following useful lemma.
\begin{lemma}\label{eq:l7}
Suppose Assumptions \ref{eq:a1}, \ref{eq:a2} and \ref{eq:a3}  hold and fix $T>0$. Then for any $\epsilon\in (0,1)$, there exists a pair $n=n(\epsilon)>0$ and $\bar{\Delta}=\bar{\Delta}(\epsilon)>0$ such that 
\begin{equation}\label{eq:33}
  \mathbb{P}(\vartheta_{\Delta,n}\leq T)\leq \epsilon
\end{equation}
as long as $\Delta \in (0, \bar \Delta]$, where
\begin{equation}\label{eq:33a}
 \vartheta_{\Delta,n}=\inf\{t\in [0,T]:x_{\Delta}(t)\notin (1/n,n)\}
\end{equation}
is a stopping time. 
\end{lemma}
\begin{proof}
Let $Z(\cdot)$ be the Lyapunov function in \eqref{eq:13}. Then for $t\in [0,T]$, the It\^{o} formula gives us
\begin{align*}
&\mathbb{E}(Z(x_{\Delta}(t\wedge \vartheta_{\Delta,n}))-Z(\xi(0)))\\
&=\mathbb{E}\int_{0}^{t\wedge\vartheta_{\Delta,n}}\Big[Z_x(x_{\Delta}(s^-))f_{\Delta}(\bar{x}_{\Delta}(s^-))
+\frac{1}{2}Z_{xx}(x_{\Delta}(s^-))\varphi(\bar{x}_{\Delta}((s-\tau)^-))^2g_{\Delta}(\bar{x}_{\Delta}(s^-))^2\\
&+\lambda(Z(x_{\Delta}(s^-)+h(\bar{x}_{\Delta}(s^-)))-Z(x_{\Delta}(s^-)))\Big]ds.
\end{align*}
By expansion, we obtain
\begin{align*}
&\mathbb{E}(Z(x_{\Delta}(t\wedge \vartheta_{\Delta,n}))-Z(\xi(0)))\\
&\leq \mathbb{E}\int_{0}^{t\wedge\vartheta_{\Delta,n}} \Big[\Big(Z_x(x_{\Delta}(s^-))f_{\Delta}(x_{\Delta}(s^-))+\frac{1}{2}Z_{xx}(x_{\Delta}(s^-))\varphi(x_{\Delta}((s-\tau)^-))^2g_{\Delta}(x_{\Delta}(s^-))^2\\
&+\lambda(Z(x_{\Delta}(s^-)+h(x_{\Delta}(s^-)))-Z(x_{\Delta}(s^-)))\Big)+ Z_x(x_{\Delta}(s^-))\Big(f_{\Delta}(\bar{x}_{\Delta}(s^-))-f_{\Delta}(x_{\Delta}(s^-))\Big)\\
&+\frac{1}{2}Z_{xx}(x_{\Delta}(s^-))\Big(\varphi(\bar{x}_{\Delta}((s-\tau)^-))^2g_{\Delta}(\bar{x}_{\Delta}(s^-))^2-\varphi(x_{\Delta}((s-\tau)^-))^2g_{\Delta}(x_{\Delta}(s^-))^2\Big)\\
&+\lambda\Big(Z(x_{\Delta}(s^-)+h(\bar{x}_{\Delta}(s^-)))-Z(x_{\Delta}(s^-)+h(x_{\Delta}(s^-)))\Big)\Big]ds\\
&\leq  \mathbb{E}\int_{0}^{t\wedge\vartheta_{\Delta,n}} L(x_{\Delta}(s^-),x_{\Delta}((s-\tau)^-))ds + H_{21}+ H_{22}+ H_{23}
\end{align*}
 Here,
\begin{equation*}
 L(x_{\Delta}(s^-),x_{\Delta}((s-\tau)^-))\leq  \ell(x_{\Delta}(s^-),x_{\Delta}((s-\tau)^-))+\lambda(Z(x_{\Delta}(s^-)+h(x_{\Delta}(s^-)))-Z(x_{\Delta}(s^-)))
\end{equation*}
is the operator \eqref{eq:4} which is independent of $t$ with
\begin{equation*}
 \ell(x_{\Delta}(s^-),x_{\Delta}((s-\tau)^-))=Z_x(x_{\Delta}(s^-))f_{\Delta}(x_{\Delta}(s^-))+\frac{1}{2}Z_{xx}(x_{\Delta}(s^-))\varphi(x_{\Delta}((s-\tau)^-))^2g_{\Delta}(x_{\Delta}(s^-))^2,
\end{equation*}
and
\begin{align*}
H_{21}&=\mathbb{E}\int_{0}^{t\wedge\vartheta_{\Delta,n}} Z_x(x_{\Delta}(s^-))\Big(f_{\Delta}(\bar{x}_{\Delta}(s^-))-f_{\Delta}(x_{\Delta}(s^-))\Big)ds\\
H_{22}&=\frac{1}{2}\mathbb{E}\int_{0}^{t\wedge\vartheta_{\Delta,n}} Z_{xx}(x_{\Delta}(s^-))\Big(\varphi(\bar{x}_{\Delta}((s-\tau)^-))^2g_{\Delta}(\bar{x}_{\Delta}(s^-))^2-\varphi(x_{\Delta}((s-\tau)^-))^2g_{\Delta}(x_{\Delta}(s^-))^2\Big)ds\\
H_{23}&=\lambda\mathbb{E}\int_{0}^{t\wedge\vartheta_{\Delta,n}} \Big(Z(x_{\Delta}(s^-)+h(\bar{x}_{\Delta}(s^-)))-Z(x_{\Delta}(s^-)+h(x_{\Delta}(s^-)))\Big)ds.
\end{align*}
By Assumption \ref{eq:a2}, there exists a constant $K_6>0$ such that for $s\in [0,t\wedge \vartheta_{\Delta,n}]$
\begin{equation*}
L(x_{\Delta}(s^-),x_{\Delta}((s-\tau)^-))\leq K_6.
\end{equation*}
Also by Lemma \ref{eq:l2}, we have
\begin{equation*}
H_{21}\leq K_n\mathbb{E}\int_{0}^{t\wedge\vartheta_{\Delta,n}} Z_x(x_{\Delta}(s^-))|\bar{x}_{\Delta}(s^-)-x_{\Delta}(s^-)|ds.
\end{equation*}
Meanwhile, for $x_{\Delta}(s^-), \bar{x}_{\Delta}(s^-)\in [1/n,n]$, we derive that
\begin{align*}
H_{22}&= \frac{1}{2}\mathbb{E}\int_{0}^{t\wedge\vartheta_{\Delta,n}} Z_{xx}\Big(g_{\Delta}(x_{\Delta}(s^-))^2|\varphi(\bar{x}_{\Delta}((s-\tau)^-))^2- \varphi(x_{\Delta}((s-\tau)^-))^2|\\
&+\varphi(\bar{x}_{\Delta}((s-\tau)^-))^2|g_{\Delta}(\bar{x}_{\Delta}(s^-))^2-g_{\Delta}(x_{\Delta}(s^-))^2|\Big)ds\\
&\leq \mathbb{E}\int_{0}^{t\wedge\vartheta_{\Delta,n}} Z_{xx}\Big(x_{\Delta}(s^-))(\sigma^2\mu(n) K_n |\bar{x}_{\Delta}(s^-)-x_{\Delta}(s^-)|\\
&+\sigma (\mu(n))^2L_n|\bar{x}_{\Delta}((s-\tau)^-)-x_{\Delta}((s-\tau)^-)|\Big)ds,
\end{align*}
where \eqref{eq:8}, \eqref{eq:18} and \eqref{eq:21} have been used. Moreover, by the definition of \eqref{eq:13}, we have
\begin{align*}
H_{23}\leq \lambda\mathbb{E}&\int_{0}^{t\wedge\vartheta_{\Delta,n}}\Big( (x_{\Delta}(s^-)+h(\bar{x}_{\Delta}(s^-)))^{\beta}-1-\beta\log(x_{\Delta}(s^-)+h(\bar{x}_{\Delta}(s^-)))\\
&-(x_{\Delta}(s^-)+h(x_{\Delta}(s^-)))^{\beta}+1+\beta\log(x_{\Delta}(s^-)+h(x_{\Delta}(s^-)))\Big)ds\\
&\le H_{231}+H_{232},
\end{align*}
where
\begin{align*}
H_{231}&=\lambda\mathbb{E}\int_{0}^{t\wedge\vartheta_{\Delta,n}}|(x_{\Delta}(s^-)+\alpha_3\bar{x}_{\Delta}(s^-))^{\beta}-(x_{\Delta}(s^-)+\alpha_3x_{\Delta}(s^-))^{\beta}|ds
\end{align*}
and
\begin{align*}
H_{232}&=\lambda\beta\mathbb{E}\int_{0}^{t\wedge\vartheta_{\Delta,n}}|\log(x_{\Delta}(s^-)+\alpha_3\bar{x}_{\Delta}(s^-))-\log(x_{\Delta}(s^-)+\alpha_3x_{\Delta}(s^-))|ds.
\end{align*}
Applying the mean value theorem, we obtain 
\begin{align*}
H_{231}&\le n\lambda\mathbb{E}\int_{0}^{t\wedge\vartheta_{\Delta,n}}|x_{\Delta}(s^-)+\alpha_3\bar{x}_{\Delta}(s^-)-\alpha_3x_{\Delta}(s^-)-x_{\Delta}(s^-)|ds\\
&= n\lambda\alpha_3\mathbb{E}\int_{0}^{t\wedge\vartheta_{\Delta,n}}|\bar{x}_{\Delta}(s^-)-x_{\Delta}(s^-)|ds.
\end{align*}
Similarly, we also have
\begin{align*}
H_{232}&\le n\lambda\beta\mathbb{E}\int_{0}^{t\wedge\vartheta_{\Delta,n}}|x_{\Delta}(s^-)+\alpha_3\bar{x}_{\Delta}(s^-)-\alpha_3x_{\Delta}(s^-)-x_{\Delta}(s^-)|ds\\
&=n\lambda\alpha_3\beta\mathbb{E}\int_{0}^{t\wedge\vartheta_{\Delta,n}}|\bar{x}_{\Delta}(s^-)-x_{\Delta}(s^-)|ds.
\end{align*}
Substituting $H_{231}$ and $H_{232}$ back into $H_{23}$, we have
\begin{equation*}
H_{23}\leq n\lambda\alpha_3(1+\beta)\mathbb{E}\int_{0}^{t\wedge\vartheta_{\Delta,n}}|\bar{x}_{\Delta}(s^-)-x_{\Delta}(s^-)|ds.
\end{equation*}
We thus combine the $H_{21}$, $H_{22}$ and $H_{23}$ to have
\begin{align*}
\mathbb{E}(Z(x_{\Delta}(t\wedge \vartheta_{\Delta,n})))&\leq Z(\xi(0))+K_6T\\
&+\sigma (\mu(n))^2L_n\mathbb{E}\int_{0}^{t\wedge\vartheta_{\Delta,n}}Z_{xx}(x_{\Delta}(s^-))|\bar{x}_{\Delta}((s-\tau)^-)-x_{\Delta}((s-\tau)^-)|ds\\
& +K_n\mathbb{E}\int_{0}^{t\wedge\vartheta_{\Delta,n}} Z_x(x_{\Delta}(s^-))|\bar{x}_{\Delta}(s^-)-x_{\Delta}(s^-)|ds\\
&+\sigma^2\mu(n) K_n \mathbb{E}\int_{0}^{t\wedge\vartheta_{\Delta,n}}Z_{xx}(x_{\Delta}(s^-))|\bar{x}_{\Delta}(s^-)-x_{\Delta}(s^-)|ds\\
&+n\lambda\alpha_3(1+\beta)\mathbb{E}\int_{0}^{t\wedge\vartheta_{\Delta,n}}|\bar{x}_{\Delta}(s^-)-x_{\Delta}(s^-)|ds.
\end{align*}
Therefore
\begin{align*}
\mathbb{E}(Z(x_{\Delta}(t\wedge \vartheta_{\Delta,n})))&\leq Z(\xi(0))+ K_6T+ K_7\mathbb{E}\int_{0}^{t\wedge\vartheta_{\Delta,n}}|\bar{x}_{\Delta}(s-\tau)-x_{\Delta}(s-\tau)|ds\\
&+K_8\mathbb{E}\int_{0}^{t\wedge\vartheta_{\Delta,n}}|\bar{x}_{\Delta}(s)-x_{\Delta}(s)|ds\\
&\leq Z(\xi(0))+ K_6T+ K_7\mathbb{E}\int_{-\tau}^{0}|\xi([s/\Delta]\Delta)-\xi(s)|ds\\
&+(K_7+K_8)\int_{0}^T\mathbb{E}\Big(\mathbb{E}|\bar{x}_{\Delta}(s)-x_{\Delta}(s)|^p\Big|\mathcal{F}_{k(s)}\Big)^{1/p}ds
\end{align*}
where
\begin{align*}
K_7&=\max_{1/n\leq x \leq n}\{Z_{xx}(x)\sigma (\mu(n))^2L_n\}
\end{align*}
and
\begin{align*}
K_8&=\max_{1/n\leq x \leq n}\{Z_x(x)K_n+Z_{xx}(x)\sigma^2\mu(n) K_n+n\lambda\alpha_3(1+\beta)\}.
\end{align*}
By Lemma \ref{eq:l5} and \ref{eq:l6}, we now have
\begin{align*}
\mathbb{E}(Z(x_{\Delta}(t\wedge \vartheta_{\Delta,n})))&\leq Z(\xi(0))+ K_6T+ K_3K_7T\Delta^{\gamma}+(K_7+K_8)\mathfrak{D}_1^{1/p}\\
&\times \Big(\Delta^{p/2}(\pi(\Delta))^p+\Delta\Big)^{1/p}\int_{0}^{T}(\mathbb{E}|\bar{x}_{\Delta}(s)|^p)^{1/p}ds\\
&\leq Z(\xi(0))+ K_6T+ K_3K_7T\Delta^{\gamma}+(K_7+K_8)\mathfrak{D}_1^{1/p}\\
&\times \Big(\Delta^{p/2}(\pi(\Delta))^p+\Delta\Big)^{1/p}\int_{0}^{T}(\sup_{0\leq u \leq s}(\mathbb{E}|\bar{x}_{\Delta}(u)|^p))^{1/p}ds\\
&\leq Z(\xi(0))+ K_6T+\nu_1\Delta^{\gamma}+\nu_2(\Delta^{p/2}(\pi(\Delta))^p+\Delta)^{1/p}\rho_3^{1/p}T.
\end{align*}
where $\nu_1= K_3K_7T$ and $\nu_2=(K_7+K_8)\mathfrak{D}_1^{1/p}$. Hence,
\begin{equation}\label{eq:34}
\mathbb{P}(\vartheta_{\Delta,n}\leq T)\leq \frac{Z(\xi(0))+ K_6T+ \nu_1\Delta^{\gamma}+\nu_2(\Delta^{p/2}(\pi(\Delta))^p+\Delta)^{1/p}\rho_3^{1/p} T}{Z(1/n)\wedge Z(n)}.
\end{equation}
For any $\epsilon\in(0,1)$, we may select sufficiently large $n$ such that
\begin{equation}\label{eq:35}
\frac{Z(\xi(0))+ K_6T}{Z(1/n)\wedge Z(n)}\leq \frac{\epsilon}{2}
\end{equation}
and sufficiently small of each step size $\Delta\in (0,\bar{\Delta}]$ such that
\begin{equation}\label{eq:36}
\frac{\nu_1\Delta^{\gamma}+\nu_2(\Delta^{p/2}(\pi(\Delta))^p+\Delta)^{1/p}\rho_3^{1/p}T}{Z(1/n)\wedge Z(n)}\leq \frac{\epsilon}{2}.
\end{equation}
Therefore, we obtain \eqref{eq:33} by combining \eqref{eq:35} and \eqref{eq:36}.
\end{proof}
We can now reveal finite time strong convergence theory of the truncated EM scheme.
\begin{lemma}\label{eq:l8}
Let Assumptions \ref{eq:a1}, \ref{eq:a2}, \ref{eq:a3} and \ref{eq:a5} hold. Set
\begin{equation*}
 \varsigma_{\Delta,n}=\vartheta_{\Delta,n} \wedge \tau_n,
\end{equation*}
where $\vartheta_{\Delta,n}$ and $\tau_n$ are \eqref{eq:12} and \eqref{eq:33a} respectively. Then for any $p\geq 2$, $T> 0$, we have
\begin{equation}\label{eq:37}
\mathbb{E}\Big( \sup_{0\leq t \leq T}|x_{\Delta}(t \wedge \varsigma_{\Delta,n})-x(t \wedge \varsigma_{\Delta,n})|^p  \Big)\leq \mathcal{K}\Delta^{p(1/4\wedge \gamma \wedge 1/p )} 
\end{equation}
for any sufficiently large $n$ and any $\Delta\in (0,\Delta^*]$, where $\mathcal{K}$ is a constant independent of $\Delta$. Consequently, we have
\begin{equation}\label{eq:38}
\lim_{\Delta\rightarrow 0}\mathbb{E}\Big( \sup_{0\leq t \leq T}|x_{\Delta}(t \wedge \varsigma_{\Delta,n})-x(t \wedge \varsigma_{\Delta,n})|^p \Big)=0.
\end{equation}
\end{lemma}
\begin{proof}
By elementary inequality, it follows from \eqref{eq:3} and \eqref{eq:30} that for $t_1\in [0,T]$
\begin{align*}
\mathbb{E}\Big( \sup_{0\leq t \leq t_1}|x_{\Delta}(t \wedge \varsigma_{\Delta,n})-x(t \wedge \varsigma_{\Delta,n})|^p\Big)&\leq H_{31}+H_{32}+H_{33}.
\end{align*}
where
\begin{align*}
&H_{31}=3^{p-1}\mathbb{E}\Big( |\int_{0}^{t_1\wedge \varsigma_{\Delta,n}}[f_{\Delta}(\bar{x}_{\Delta}(s^-))-f(x(s^-))]ds|^p\Big),\\
&H_{32}=3^{p-1} \mathbb{E}\Big( \sup_{0\leq t \leq t_1}|\int_{0}^{t\wedge \varsigma_{\Delta,n}}[\varphi(\bar{x}_{\Delta}((s-\tau)^-))g_{\Delta}(\bar{x}_{\Delta}(s^-))-\varphi(x((s-\tau)^-))g(x(s^-))]dB(s)|^p\Big)\\
\text{and}& &\\
&H_{33}= 3^{p-1}\mathbb{E}\Big( \sup_{0\leq t \leq t_1}|\int_{0}^{t\wedge \varsigma_{\Delta,n}}[h(\bar{x}_{\Delta}(s^-))-h(x(s^-))]dN(s)|^p\Big).
\end{align*}
By the H\"older inequality and Lemma \ref{eq:l2}, we have 
\begin{equation}\label{eq:39}
H_{31}\leq 3^{p-1}T^{p-1}K_n^p \mathbb{E}\int_{0}^{t_1\wedge \varsigma_{\Delta,n}}|\bar{x}_{\Delta}(s^-)-x(s^-)|^pds,
\end{equation}
Furthermore, the H\"older and Burkholder-Davis Gundy inequalities yield
\begin{align*}
H_{32}&\leq 3^{p-1}T^{\frac{p-2}{2}}c_p \mathbb{E}\int_{0}^{t_1\wedge \varsigma_{\Delta,n}}\Big(|\varphi(\bar{x}_{\Delta}((s-\tau)^-))g_{\Delta}(\bar{x}_{\Delta}(s^-))-\varphi(x((s-\tau)^-))g_{\Delta}(\bar{x}_{\Delta}(s^-))\\
&+\varphi(x((s-\tau)^-))g_{\Delta}(\bar{x}_{\Delta}(s^-))-\varphi(x((s-\tau)^-))g(x(s^-))|^p\Big)ds\\
&\leq 2^{p-1}3^{p-1}T^{\frac{p-2}{2}}c_p \mathbb{E}\int_{0}^{t_1\wedge  \varsigma_{\Delta,n}}\Big(g_{\Delta}(\bar{x}_{\Delta}(s^-))^p|\varphi(\bar{x}_{\Delta}((s-\tau)^-))-\varphi(x((s-\tau)^-))|^p\\
&+\varphi(x((s-\tau)^-))^p|g_{\Delta}(\bar{x}_{\Delta}(s^-))-g(x(s^-))|^p \Big)ds,
\end{align*}
where $c_p$ is a positive constant. For $s\in[0,t_1\wedge  \varsigma_{\Delta,n}]$, we have $x_{\Delta}(s^-),\bar{x}_{\Delta}(s^-)\in [1/n,n] $. So by Assumption \ref{eq:a3}, Lemma \ref{eq:l2} and \eqref{eq:23}, we now have
\begin{align}\label{eq:40}
H_{32}&\leq 2^{p-1}3^{p-1}T^{\frac{p-2}{2}}c_p L_n^p(\mu(n))^p\mathbb{E}\int_{-\tau}^{0}|\xi([s/\Delta]\Delta)-\xi(s)|^pds\\  \nonumber
&+2^{p-1}3^{p-1}T^{\frac{p-2}{2}}c_p(L_n^p(\mu(n))^p +K_n^p\sigma^p)\mathbb{E}\int_{0}^{t_1\wedge  \varsigma_{\Delta,n}}|\bar{x}_{\Delta}(s^-)-x(s^-)|^pds. \\ \nonumber
&\leq 2^{p-1}3^{p-1}T^{\frac{p-2}{2}}c_pL_n^p K_3^p(\mu(n))^p\Delta^{p\gamma}\tau +2^{p-1}3^{p-1}T^{\frac{p-2}{2}}c_p\Big(L_n^p(\mu(n))^p +K_n^p\sigma^p\Big)\\
&\times \mathbb{E}\int_{0}^{t_1\wedge  \varsigma_{\Delta,n}}|\bar{x}_{\Delta}(s^-)-x(s^-)|^pds.
\end{align}
Moreover, we obtain from elementary inequality
\begin{align*}
H_{33}&\leq  3^{p-1}\mathbb{E}\Big( \sup_{0\leq t \leq t_1}|\int_{0}^{t\wedge \varsigma_{\Delta,n}}[h(\bar{x}_{\Delta}(s^-))-h(x(s^-))]d\widetilde{N}(s)\\
&+ \lambda\int_{0}^{t \wedge \varsigma_{\Delta,n}}[h(\bar{x}_{\Delta}(s^-))-h(x(s^-))]ds|^p\Big)\\
&\leq H_{331}+ H_{332},
\end{align*}
where
\begin{align*}
H_{331}&=2^{p-1} 3^{p-1}\mathbb{E}\Big( \sup_{0\leq t \leq t_1}|\int_{0}^{t\wedge \varsigma_{\Delta,n}}[h(\bar{x}_{\Delta}(s^-))-h(x(s^-))]d\widetilde{N}(s)|^p\Big)
\end{align*}
and
\begin{align*}
H_{332}&=2^{p-1} 3^{p-1}\lambda^p\mathbb{E}\Big( \sup_{0\leq t \leq t_1}|\int_{0}^{t\wedge \varsigma_{\Delta,n}}[h(\bar{x}_{\Delta}(s^-))-h(x(s^-))]ds|^p\Big).
\end{align*}
The Doob martingale inequality, martingale isometry and Lemma \ref{eq:l2} give us
\begin{align*}
H_{331}&\le 2^{p-1} 3^{p-1}\bar{c}_p\lambda^{\frac{p}{2}} \Big(\mathbb{E}\int_{0}^{t_1\wedge \varsigma_{\Delta,n}} |h(\bar{x}_{\Delta}(s^-))-h(x(s^-))|^2d\widetilde{N}(s) \Big)^{\frac{p}{2}}\\
&\le 2^{p-1} 3^{p-1}\bar{c}_p \lambda^{\frac{p}{2}}T^{\frac{p-2}{2}}K_n^p\mathbb{E}\int_{0}^{t_1\wedge \varsigma_{\Delta,n}}|\bar{x}_{\Delta}(s^-)-x(s^-)|^pds,
\end{align*}
where $\bar{c}_p$ is a positive constant. Moreover by the H\"older inequality and Lemma \ref{eq:l2},
\begin{align*}
H_{332}&\le 2^{p-1} 3^{p-1}\lambda^pT^{p-1}\mathbb{E}\int_{0}^{t_1\wedge\varsigma_{\Delta,n}} |h(\bar{x}_{\Delta}(s^-))-h(x(s^-))|^pds\\
&\le 2^{p-1} 3^{p-1}\lambda^pT^{p-1}K_n^p\mathbb{E}\int_{0}^{t_1\wedge \varsigma_{\Delta,n}}|\bar{x}_{\Delta}(s^-)-x(s^-)|^pds,
\end{align*}
where Lemma \ref{eq:l2} has been used. Substituting $H_{331}$ and $H_{332}$ into $H_{33}$ yields
\begin{equation}\label{eq:41}
H_{33}\leq  2^{p-1} 3^{p-1}K_n^p(\bar{c}_p \lambda^{\frac{p}{2}}T^{\frac{p-2}{2}}+\lambda^pT^{p-1})\mathbb{E}\int_{0}^{t_1\wedge \varsigma_{\Delta,n}}|\bar{x}_{\Delta}(s^-)-x(s^-)|^pds.
\end{equation}
We now combine \eqref{eq:39}, \eqref{eq:40} and \eqref{eq:41} to have
\begin{align*}
\mathbb{E}\Big( \sup_{0\leq t \leq t_1}|x_{\Delta}(t \wedge \varsigma_{\Delta,n})-x(t \wedge \varsigma_{\Delta,n})|^p\Big)&\leq \zeta_1\Delta^{p\gamma}\tau+(\zeta_2+\zeta_3+\zeta_4)\mathbb{E}\int_{0}^{t_1\wedge \varsigma_{\Delta,n}}|\bar{x}_{\Delta}(s^-)-x(s^-)|^pds\\
&\leq \zeta_1\Delta^{p\gamma}\tau+(\zeta_2+\zeta_3+\zeta_4)\mathbb{E}\int_{0}^{t_1\wedge \varsigma_{\Delta,n}}|\bar{x}_{\Delta}(s)-x(s)|^pds,
\end{align*}
where
\begin{align*}
\zeta_1&=2^{p-1}3^{p-1}T^{\frac{p-2}{2}}c_pL_n^pK_3^p(\mu(n))^p\\
\zeta_2&=3^{p-1}T^{p-1}K_n^p\\
\zeta_3&=2^{p-1}3^{p-1}T^{\frac{p-2}{2}}c_p(L_n^p(\mu(n))^p +K_n^p\sigma^p)
\end{align*}
and 
\begin{align*}
\zeta_4&=2^{p-1}3^{p-1}K_n^p(\bar{c}_p \lambda^{\frac{p}{2}}T^{\frac{p-2}{2}}+\lambda^pT^{p-1}).
\end{align*}
Meanwhile by elementary inequality and Lemma \ref{eq:l5},
\begin{align*}
&\mathbb{E}\Big( \sup_{0\leq t \leq t_1}|x_{\Delta}(t \wedge \varsigma_{\Delta,n})-x(t \wedge \varsigma_{\Delta,n})|^p\Big)\\
&\leq \zeta_1\Delta^{p\gamma}\tau+2^{p-1}(\zeta_2+\zeta_3+\zeta_4)\int_{0}^{T}\mathbb{E}\Big(\mathbb{E}|\bar{x}_{\Delta}(s)-x_{\Delta}(s)|^p\big|\mathcal{F}_{k(s)}\Big)ds\\
&+2^{p-1}(\zeta_2+\zeta_3+\zeta_4)\int_{0}^{t_1}\mathbb{E}\Big(\sup_{0\leq t\leq s}|x_{\Delta}(t\wedge \varsigma_{\Delta,n})-x(t\wedge \varsigma_{\Delta,n})|^p\Big)ds\\
&\le \zeta_1\Delta^{p\gamma}\tau+2^{p-1}(\zeta_2+\zeta_3+\zeta_4)\mathfrak{D}_1\Big(\Delta^{p/2}(\pi(\Delta))^p+\Delta\Big)\int_{0}^{T}\mathbb{E}|\bar{x}_{\Delta}(s)|^pds\\
&+2^{p-1}(\zeta_2+\zeta_3+\zeta_4)\int_{0}^{t_1}\mathbb{E}\Big(\sup_{0\leq t\leq s}|x_{\Delta}(t\wedge \varsigma_{\Delta,n})-x(t\wedge \varsigma_{\Delta,n})|^p\Big)ds
\end{align*}
So by Lemma \ref{eq:l6}, we have 
\begin{align*}
&\mathbb{E}\Big( \sup_{0\leq t \le t_1}|x_{\Delta}(t \wedge \varsigma_{\Delta,n})-x(t \wedge \varsigma_{\Delta,n})|^p\Big)\\
&\le \zeta_1\tau\Delta^{p\gamma}+2^{p-1}\rho_3 \mathfrak{D}_1T(\zeta_2+\zeta_3+\zeta_4)\Big([\Delta^{p/4}(\pi(\Delta))^p]\Delta^{p/4}+\Delta^{p(1/p)}\Big)\\
&+2^{p-1}(\zeta_2+\zeta_3+\zeta_4)\int_{0}^{t_1}\mathbb{E}\Big(\sup_{0\leq t\leq s}|x_{\Delta}(t \wedge \varsigma_{\Delta,n})-x(t \wedge \varsigma_{\Delta,n})|^p\Big)ds\\
&\le \Big(\zeta_1\tau+2^{p-1}\rho_3 \mathfrak{D}_1T(\zeta_2+\zeta_3+\zeta_4)(\Delta^{p/4}(\pi(\Delta))^p+1)\Big)\Delta^{p(1/4  \wedge \gamma \wedge 1/p)}\\
&+2^{p-1}(\zeta_2+\zeta_3+\zeta_4)\int_{0}^{t_1}\mathbb{E}\Big(\sup_{0\leq t\leq s}|x_{\Delta}(t \wedge \varsigma_{\Delta,n})-x(t \wedge \varsigma_{\Delta,n})|^p\Big)ds.
\end{align*}
Noting from \eqref{eq:24} that $[\Delta^{1/4}(\pi(\Delta))]^{p}\le 1$, we have 
\begin{align*}
&\mathbb{E}\Big( \sup_{0\leq t \le t_1}|x_{\Delta}(t \wedge \varsigma_{\Delta,n})-x(t \wedge \varsigma_{\Delta,n})|^p\Big)\\
&\le \Big(\zeta_1\tau+2^{p}\rho_3 \mathfrak{D}_1T(\zeta_2+\zeta_3+\zeta_4)\Big)\Delta^{p(1/4  \wedge \gamma \wedge 1/p)}\\
&+2^{p-1}(\zeta_2+\zeta_3+\zeta_4)\int_{0}^{t_1}\mathbb{E}\Big(\sup_{0\leq t\leq s}|x_{\Delta}(t \wedge \varsigma_{\Delta,n})-x(t \wedge \varsigma_{\Delta,n})|^p\Big)ds.
\end{align*}
By the Gronwall inequality, we obtain
\begin{equation*}
  \mathbb{E}\Big( \sup_{0\leq t \leq T}|x_{\Delta}(t \wedge \varsigma_{\Delta,n})-x(t \wedge \varsigma_{\Delta,n})|^p\Big)\le\mathcal{K} \Delta^{p(1/4\wedge \gamma \wedge 1/p)},
\end{equation*}
where $\mathcal{K}=\varrho_1(p) e^{\varrho_2(p)}$ with 
\begin{align*}
\varrho_1(p)&=\zeta_1\tau+2^{p}\rho_3 \mathfrak{D}_1T(\zeta_2+\zeta_3+\zeta_4)
\end{align*}
and
\begin{align*} 
\varrho_2(p)&=2^{p-1}(\zeta_2+\zeta_3+\zeta_4).
\end{align*}
Moreover, the required inequality \eqref{eq:38} is deduced by setting $\Delta\rightarrow 0$.
\end{proof}
The following gives the strong convergence theory of the truncated EM scheme.
\begin{theorem}\label{eq:thrm2}
Let Assumptions \ref{eq:a1}, \ref{eq:a2}, \ref{eq:a3} and \ref{eq:a5} hold. Then for any $p\ge 2$, we have
\begin{equation}\label{eq:42}
\lim_{\Delta\rightarrow 0}\mathbb{E}\Big( \sup_{0\leq t \leq T}|x_{\Delta}(t)-x(t)|^p \Big)=0
\end{equation}
and consequently
\begin{equation}\label{eq:43}
\lim_{\Delta\rightarrow 0}\mathbb{E}\Big( \sup_{0\leq t \leq T}|\bar{x}_{\Delta}(t)-x(t)|^p \Big)=0.
\end{equation}
\end{theorem}
\begin{proof}
We only need to prove the theorem for $ p\ge 3$ as for $p\in [2,3)$ it follows from the case of $p=3$ and the H\"older inequality. Let $\vartheta_{\Delta,n}$, $\tau_n$ and $\varsigma_{\Delta,n}$ be the same as before and set
\begin{equation*}
e_{\Delta}(t)=x_{\Delta}(t)-x(t).
\end{equation*}
For any arbitrarily $ \delta >0$, the Young inequality 
\begin{equation*}
\mathfrak{a}\mathfrak{b}\le \frac{\delta}{2}\mathfrak{a}^2+\frac{1}{2\delta}\mathfrak{b}^2
\quad \forall \mathfrak{a},\mathfrak{b} > 0,
\end{equation*}
yields
\begin{align}\label{eq:44}
\mathbb{E}\Big(\sup_{0\leq t \leq T}|e_{\Delta}(t)|^p\Big)&=\mathbb{E}\Big(\sup_{0\leq t \leq T}|e_{\Delta}(t)|^p1_{\{\tau_n>T \text{ and }\vartheta_{\Delta,n}>T\}}\Big)+\mathbb{E}\Big(\sup_{0\leq t \leq T}|e_{\Delta}(t)|^p1_{\{\tau_n\le T \text{ or }\vartheta_{\Delta,n}\le T\}}\Big)\nonumber\\
&\leq \mathbb{E}\Big(\sup_{0\leq t \leq T}|e_{\Delta}(t)|^p1_{\{\varsigma_{\Delta,n}>T\}}\Big)+\frac{\delta}{2}\mathbb{E}\Big(\sup_{0\leq t \leq T}|e_{\Delta}(t)|^{2p}\Big)\nonumber\\
&+\frac{1}{2\delta}\mathbb{P}(\tau_n\leq T \text{ or }\vartheta_{\Delta,n} \leq T).
\end{align}
So for $ p\ge 3$, Lemmas \ref{eq:l1} and \ref{eq:l6} give us
\begin{align}\label{eq:45}
 \mathbb{E}\Big(\sup_{0\leq t \leq T} |e_{\Delta}(t)|^{2p}\Big)&\leq 2^{2p}\mathbb{E}\Big(\sup_{0\leq t \leq T}|x(t)|^{2p}\vee \sup_{0\leq t \leq T}|x_{\Delta}(t)|^{2p}\Big)
 \nonumber\\
&\le 2^{2p}(\rho_1\vee \rho_3)^2.
\end{align}
Also by Theorem \ref{eq:thrm1} and Lemma \ref{eq:l8},
\begin{equation}\label{eq:46}
\mathbb{P}(\varsigma_{\Delta,n}\leq T)\leq \mathbb{P}(\tau_n\leq T) +\mathbb{P}(\vartheta_{\Delta,n}\leq T).
\end{equation}
Moreover, by Lemma \ref{eq:l8}
\begin{equation}\label{eq:47}
\mathbb{E}\Big(\sup_{0\leq t \leq T}|e_{\Delta}(t)|^p1_{\{\varsigma_{\Delta,n}>T\}}\Big)\leq \mathcal{K}\Delta^{p(1/4\wedge \gamma \wedge 1/p)}.
\end{equation}
Therefore, we substitute \eqref{eq:45}, \eqref{eq:46} and \eqref{eq:47} into \eqref{eq:44} to have
\begin{align*}
\mathbb{E}\Big(\sup_{0\leq t \leq T}|e_{\Delta}(t)|^p\Big)&\leq \frac{2^{2p}(\rho_1\vee \rho_3)^2\delta}{2}+\mathcal{K}\Delta^{p(1/4\wedge \gamma \wedge 1/p)}+\frac{1}{2\delta}\mathbb{P}(\tau_n\leq T)+\frac{1}{2\delta}\mathbb{P}(\vartheta_{\Delta,n}\leq T).
\end{align*}
Given $\epsilon\in (0,1)$, we can select $\delta$ so that
\begin{equation}\label{eq:48}
\frac{2^{2p}(\rho_1\vee \rho_3)^2\delta}{2} \leq\frac{\epsilon}{4}.
\end{equation}
Similarly, for any given $\epsilon\in (0,1)$, there exists $n_o$ so that for $n\geq n_o$, we may select $\delta$ to have
\begin{equation}\label{eq:49}
\frac{1}{2\delta}\mathbb{P}(\tau_n\leq T)\leq \frac{\epsilon}{4}
\end{equation}
and select $n(\epsilon)\leq n_o$ such that for $\Delta\in (0,\bar{\Delta}]$
\begin{equation}\label{eq:50}
\frac{1}{2\delta}\mathbb{P}(\vartheta_{\Delta,n}\leq T)\leq \frac{\epsilon}{4}.
\end{equation}
Finally, we may select $\Delta\in (0,\bar{\Delta}]$ sufficiently small for $\epsilon\in (0,1)$ such that
\begin{equation}\label{eq:51}
\mathcal{K}\Delta^{p(1/4\wedge \gamma \wedge 1/p)}\le \frac{\epsilon}{4}.
\end{equation}
Combining \eqref{eq:48}, \eqref{eq:49}, \eqref{eq:50} and \eqref{eq:51}, we establish
\begin{align*}
\mathbb{E}\Big(\sup_{0\leq t \leq T}|x_{\Delta}(t)-x(t)|^p\Big)\leq \epsilon.
\end{align*}
Therefore, we obtain \eqref{eq:42} and clearly, by Lemma \ref{eq:l5}, also get \eqref{eq:43} by letting $\Delta\rightarrow 0$.
\end{proof}
\section{Numerical examples}
In this section, we analyse the strong convergence result established in Theorem \ref{eq:thrm2} by comparing the truncated Euler-Maruyama (TEM) scheme with backward Euler-Maruyama (BEM) scheme for SDDE \eqref{eq:4} without $\alpha_{-1}x(t)^{-1}$ term. It is already noted in \cite{emma} that the BEM scheme is not known to cope with $\alpha_{-1}x(t)^{-1}$ term. Consider the following form of SDDE \eqref{eq:3}
\begin{equation}\label{eq:sm1}
 dx(t)=(\alpha_{-1}x(t^-)^{-1}-\alpha_{0}+\alpha_{1}x(t^-)-\alpha_{2}x(t^-)^{2})dt+\varphi(x((t-1)^-))x(t^-)^{5/4}dB(t)+\alpha_{3}x(t^{-})dN(t),
\end{equation}
with initial data $\xi(t)=0.2$. Here $\varphi(\cdot)$ is a sigmoid-type function defined by
\begin{equation}\label{eq:sm2}
\varphi(y)=
\begin{cases}
  \frac{1}{2}\frac{(1+(e^{y}-e^{-y}))}{(e^{y}+e^{-y})}, & \mbox{if $y\geq 0$ }\\
  \frac{1}{4},                             & \mbox{Otherwise},
\end{cases}
\end{equation}
Obviously, equation \eqref{eq:sm2} meets all the conditions imposed on $\varphi(\cdot)$ (see \cite{emma}). The coefficient terms $f(x)=\alpha_{-1}x^{-1}-\alpha_{0}+\alpha_{1}x-\alpha_{2}x^{2}$ and $g(x)=x^{5/4}$ are locally Lipschitz continuous and hence fulfil Assumption \ref{eq:l1}. Moreover, we easily observe
\begin{equation*}
\sup_{1/u \leq x\leq u}(|f(x)|\vee g(x))\leq K_9u^2,\quad u\ge 1,
\end{equation*}
where $K_9=\alpha_{-1}+\alpha_{0}+\alpha_{1}+ \alpha_{2}+\alpha_{3}$. We can now use $\mu=K_9u^2$ with inverse $\mu^{-1}(u)=(u/K_9)^{1/2}$.
\subsection{Numerical results}
By choosing $\pi(\Delta)=\Delta^{-2/3}$, step size $10^{-3}$ and the following coefficient values in Table \ref{Tab:table1},  we obtain Monte Carlo simulated sample path of $x(t)$ to \textup{SDDE} \eqref{eq:sm1} in Figure \ref{Fig:figure1} using the TEM scheme.
\begin{table}[!htbp]
 \centering
\begin{tabular}{ |c|c|c|c|c|}
  \hline
  $\alpha_{-1}$ & $\alpha_{0}$ & $\alpha_{1}$ & $\alpha_{2}$&$\alpha_{3}$\\
  \hline 
  0.2  & 0.3  & 0.2 & 0.5 & 1 \\
  \hline
\end{tabular}
\caption{Coefficient values including $\alpha_{-1}$}
\label{Tab:table1}
\end{table}
\par
By similarly choosing $\pi(\Delta)=\Delta^{-2/3}$, step size $10^{-3}$ and the coefficient values in Table \ref{Tab:table2} below, we also obtain Monte Carlo simulated sample paths of $x(t)$ to \textup{SDDE} \eqref{eq:sm1} in Figure \ref{Fig:figure2} using the TEM and BEM schemes. We can clearly see the TEM scheme converges strongly to BEM scheme. 
\begin{table}[!htbp]
 \centering
\begin{tabular}{ |c|c|c|c|}
  \hline
  $\alpha_{0}$ & $\alpha_{1}$ & $\alpha_{2}$&$\alpha_{3}$\\
  \hline 
  0.3  & 0.2 & 0.5 & 1 \\
  \hline
\end{tabular}
\caption{Coefficient values excluding $\alpha_{-1}$}
\label{Tab:table2}
\end{table}
\par
Finally, the plot of strong errors between these two numerical schemes is displayed in Figure \ref{Fig:figure3} with reference line at 0. Do note this simulated result of strong errors is not yet established theoretically.
\begin{figure}[!htbp]
  \centerline{\includegraphics[scale=1]{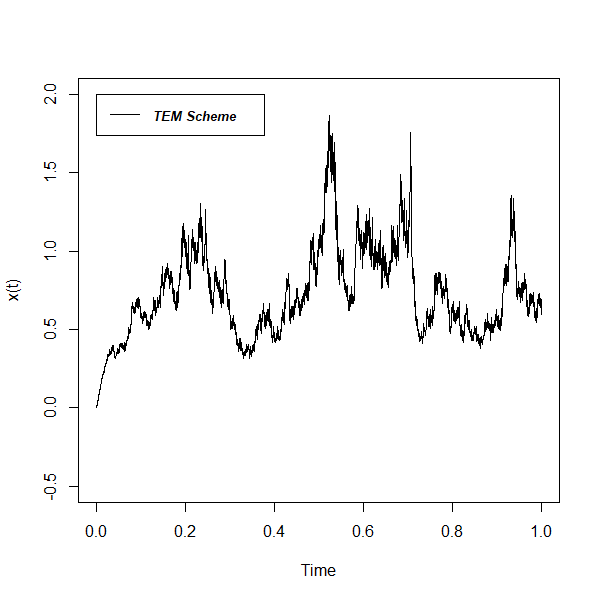}}
  \caption{Simulated sample path of $x(t)$ using $\Delta=0.001$}
  \label{Fig:figure1}
\end{figure}
\begin{figure}[!htbp]
  \centerline{\includegraphics[scale=1]{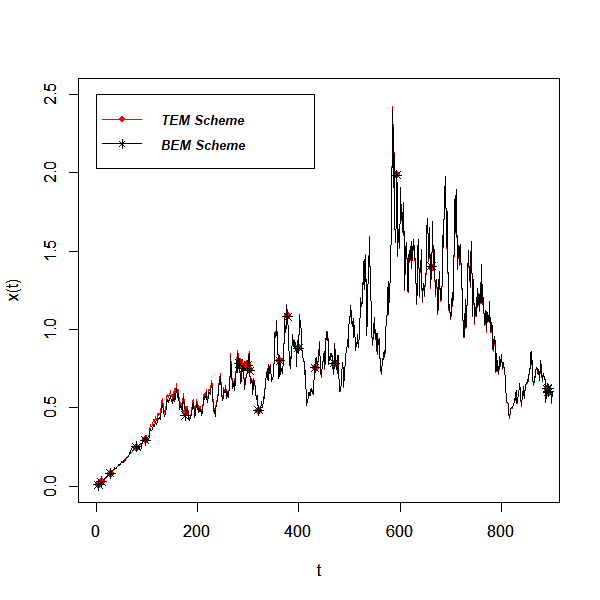}}
  \caption{Simulated sample paths of $x(t)$ using $\Delta=0.0001$}
  \label{Fig:figure2}
\end{figure}
\begin{figure}[!htbp]
  \centerline{\includegraphics[scale=1]{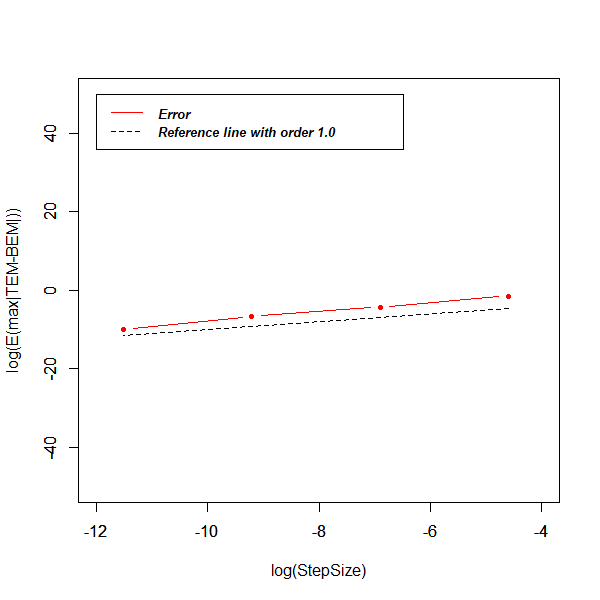}}
  \caption{Strong errors between the TEM and BEM schemes}
  \label{Fig:figure3}
\end{figure}
\section{Financial applications}
In this section, we apply Theorem \ref{eq:thrm2} to justify numerical schemes within Monte Carlo framework that value the expected payoffs of a bond and a barrier option.
\subsection{A bond}
\noindent Suppose the short-term interest rate is explained by \eqref{eq:3}. Then a bond price $P_{\mathbf{B}}$ at maturity time $T$ is of the form 
\begin{equation}\label{eq:52}
  P_{\mathbf{B}}(T)=\mathbb{E}\Big[ \exp \Big(-\int_{0}^{T}x(t)dt \Big)\Big].
\end{equation}
The approximate value of \eqref{eq:52} using \eqref{eq:28} is computed by
\begin{equation*}
  P_{\mathbf{B}\Delta}(T)=\mathbb{E}\Big[ \exp \Big(-\int_{0}^{T}\bar{x}_{\Delta}(t)dt \Big)\Big].
\end{equation*}
Hence by Theorem \ref{eq:thrm2}
\begin{equation*}
  \lim_{\Delta\rightarrow 0}|P_{\mathbf{B}\Delta}(T)-P_{\mathbf{B}}(T)|=0.
\end{equation*}
\subsection{\textbf{A barrier option}}
Consider the payoff of a path-dependent barrier option at an expiry date $T$ given by
\begin{equation}\label{eq:53}
  P_{O}(T)=\mathbb{E}\Big[ (x(T)-\mathfrak{E})^+1_{\sup_{0\leq t\leq T }}x(t)<\mathfrak{B})\Big],
\end{equation}
where the barrier,  $\mathfrak{B}$, and exercise price, $\mathfrak{E}$, are constants. Then the approximate value of \eqref{eq:53} using \eqref{eq:28} is computed by
\begin{equation*}
  P_{O\Delta}(T)=\mathbb{E}\Big[ (\bar{x}_{\Delta}(T)-\mathfrak{E})^+1_{\sup_{0\leq t\leq T }}\bar{x}_{\Delta}(t)<\mathfrak{B})\Big].
\end{equation*}
Similarly, by Theorem \ref{eq:thrm2}
\begin{equation*}
  \lim_{\Delta\rightarrow 0}|P_{O\Delta}(T)-P_{O}(T)|=0.
\end{equation*}
See, for instance, \cite{highamao} for detailed accounts.
\section*{\textit{Acknowledgements}}
The Author would like to thank University of Strathclyde for the scholarship and his first PhD supervisor, Prof. Xuerong Mao.


\begin{thebibliography}{1}
\bibitem{blackshole}
Black, F. and Scholes, M., 1973. The pricing of options and corporate liabilities. Journal of political economy, 81(3), pp.637-654.
\bibitem{merton}
Merton, R.C., 1973. Theory of rational option pricing. Theory of Valuation, pp.229-288.
\bibitem{vasicek}
Vasicek, O., 1977. An equilibrium characterization of the term structure. Journal of financial economics, 5(2), pp.177-188.
\bibitem{dothan}
Dothan, L.U., 1978. On the term structure of interest rates. Journal of Financial Economics, 6(1), pp.59-69.
\bibitem{brennan}
Brennan, M.J. and Schwartz, E.S., 1980. Analyzing convertible bonds. Journal of Financial and Quantitative analysis, 15(4), pp.907-929.
\bibitem{cox3}
Cox, J.C., Ingersoll Jr, J.E. and Ross, S.A., 2005. A theory of the term structure of interest rates. In Theory of Valuation (pp. 129-164).
\bibitem{mao3}
Mao, X. and Sabanis, S., 2013. Delay geometric Brownian motion in financial option valuation. Stochastics An International Journal of Probability and Stochastic Processes, 85(2), pp.295-320.
\bibitem{jmerton}
Merton, R.C., 1976. Option pricing when underlying stock returns are discontinuous. Journal of financial economics, 3(1-2), pp.125-144.
\bibitem{linyeh}
Lin, B.H. and Yeh, S.K., 1999. Jump‐Diffusion Interest Rate Process: An Empirical Examination. Journal of Business Finance and Accounting, 26(7‐8), pp.967-995.
\bibitem{kou}
Kou, S.G., 2002. A jump-diffusion model for option pricing. Management science, 48(8), pp.1086-1101.
\bibitem{wu2}
Wu, F., Mao, X. and Chen, K., 2008. Strong convergence of Monte Carlo simulations of the mean-reverting square root process with jump. Applied Mathematics and Computation, 206(1), pp.494-505.
\bibitem{wu3}
Wu, F., Mao, X. and Chen, K., 2009. The Cox–Ingersoll–Ross model with delay and strong convergence of its Euler–Maruyama approximate solutions. Applied Numerical Mathematics, 59(10), pp.2641-2658.
\bibitem{cev}
Lee, M.K. and Kim, J.H., 2016. A delayed stochastic volatility correction to the constant elasticity of variance model. Acta Mathematicae Applicatae Sinica, English Series, 32(3), pp.611-622.
\bibitem{aitsahalia}
Ait-Sahalia, Y., 1996. Testing continuous-time models of the spot interest rate. The review of financial studies, 9(2), pp.385-426.
\bibitem{Szpruch}
Szpruch, L., Mao, X., Higham, D.J. and Pan, J., 2011. Numerical simulation of a strongly nonlinear Ait-Sahalia-type interest rate model. BIT Numerical Mathematics, 51(2), pp.405-425.
\bibitem{cheng}
Cheng, S.R., 2009. Highly nonlinear model in finance and convergence of Monte Carlo simulations. Journal of Mathematical Analysis and Applications, 353(2), pp.531-543.
\bibitem{dung}
Dung, N.T., 2016. Tail probabilities of solutions to a generalized Ait-Sahalia interest rate model. Statistics and Probability Letters, 112, pp.98-104.
\bibitem{Deng1}
Deng, S., Fei, C., Fei, W. and Mao, X., 2019. Generalized Ait-Sahalia-type interest rate model with Poisson jumps and convergence of the numerical approximation. Physica A: Statistical Mechanics and its Applications, 533, p.122057.
\bibitem{lewis}
Lewis, A.L., 2009. Option Valuation Under Stochastic Volatility II. Finance Press, Newport Beach, CA.
\bibitem{SDEPoisson}
Higham, D.J. and Kloeden, P.E., 2005. Numerical methods for nonlinear stochastic differential equations with jumps. Numerische Mathematik, 101(1), pp.101-119.
\bibitem{maobook}
Mao, X., 2007. Stochastic differential equations and applications. 2nd ed. Chichester: Horwood Publishing Limited.
\bibitem{highamao}
Higham, D.J. and Mao, X., 2005. Convergence of Monte Carlo simulations involving the mean-reverting square root process. Journal of Computational Finance, 8(3), pp.35-61.
\bibitem{mao4}
Mao, X., 2015. The truncated Euler–Maruyama method for stochastic differential equations. Journal of Computational and Applied Mathematics, 290, pp.370-384.
\bibitem{emma}
Coffie, E. and Mao, X., Truncated EM numerical method for generalised Ait-Sahalia-type interest rate model with delay. Journal of Computational and Applied Mathematics, 383.
\bibitem{guo}
Guo, Q., Mao, X. and Yue, R., 2018. The truncated Euler–Maruyama method for stochastic differential delay equations. Numerical Algorithms, 78(2), pp.599-624.
\bibitem{Deng}
Deng, S., Fei, W., Liu, W. and Mao, X., 2019. The truncated EM method for stochastic differential equations with Poisson jumps. Journal of Computational and Applied Mathematics, 355, pp.232-257.
\bibitem{HiKlo}
Bao, J., Böttcher, B., Mao, X. and Yuan, C., 2011. Convergence rate of numerical solutions to SFDEs with jumps. Journal of Computational and Applied Mathematics, 236(2), pp.119-131.
\bibitem{espen}
Benth, F.E., 2003. Option theory with stochastic analysis: an introduction to mathematical finance. Berlin: Springer Science and Business Media.
\bibitem{StoPoisson}
Applebaum, D., 2009. Lévy processes and stochastic calculus. Cambridge university press.
\end{thebibliography}
\end{document}